\documentclass{article}
\usepackage[dvips]{graphicx}
\usepackage{amsmath}
\usepackage{amsfonts}
\usepackage{amssymb}
\usepackage{amsthm}
\usepackage{mathtools}

\usepackage{multirow}
\usepackage{array}
\usepackage{tabularx}
\usepackage{mathrsfs}
\usepackage{url}

\mathtoolsset{showonlyrefs=true}

\newtheorem{theorem}{Theorem}[section]
\newtheorem{assumption}{Assumption}[section]
\newtheorem{lemma}{Lemma}[section]

\newtheorem{proposition}{Proposition}[section]

\theoremstyle{definition}
\newtheorem{remark}{Remark}[section]
\newtheorem{definition}{Definition}[section]
\newtheorem{example}{Example}[section]

\makeatletter

\@addtoreset{equation}{section}
\makeatother

\begin{document}

\title{Contingent Claim Valuation under
Increasing Profit, Strong Arbitrage, and Arbitrage of the First Kind\thanks{We would like to thank Professor Yoshiki Otobe for helpful discussions. This study was supported by JSPS KAKENHI, Grant Number 23K01466.}}
\author{Yukihiro Tsuzuki\thanks{Faculty of Economics and Law, Shinshu University, 3-1-1 Asahi, Matsumoto, Nagano 390-8621, Japan; yukihirotsuzuki@shinshu-u.ac.jp}}

\maketitle

\abstract{
We study the upper hedging price for contingent claims in market models with strong types of arbitrage: increasing profit, strong arbitrage, and arbitrage of the first kind.
The existence of arbitrage may make the price smaller than if it did not exist.
For example, when the asset price process has a reflecting boundary,
which introduces increasing profit in the market model,
the option prices are reduced to those of the corresponding options that knock-out at the boundary.
Furthermore, we demonstrate that corporate stock price processes with increasing profit are obtained as a result of corporate stock issuance and repurchase plans.
}

\clearpage
\section{Introduction}
\label{sec:introduction}
Contingent claim valuation is one of the most fundamental questions in mathematical finance.
This question has been studied under several conditions to exclude arbitrage opportunities.
Here, an arbitrage opportunity is the possibility of making a profit without risk.
Various forms of arbitrage have been proposed in the literature
and problems in mathematical finance has been considered under strong no-arbitrage conditions and then weaker ones.

We shortly review the literature on no-arbitrage theory (see Fontana \cite{Fontana2015}
for a unified analysis of no-arbitrage conditions for continuous semimartingale models).
The seminal study of Delbaen and Schachermayer \cite{Delbaen1994} established the equivalence between the {\it No Free Lunch with Vanishing Risk (NFLVR)} condition
and the existence of an {\it Equivalent Local Martingale Measure (ELMM)}.
Notwithstanding the significance of this result, more recently, the NFLVR condition has been considered as rather strong. 
This condition is not necessarily assumed for contingent claim valuation, portfolio optimization, and so on,
and weaker no-arbitrage conditions had been proposed.
For example, {\it Stochastic Portfolio Theory} is a framework for analyzing portfolio behavior and equity market structure,
without postulating the existence of ELMMs (see Karatzas and Fernholz \cite{KaratzasFernholz2009} for an overview).
The {\it Benchmark Approach} developed a theory of contingent claim valuation, 
taking the growth optimal portfolio as the num\'{e}raire or benchmark,
also without relying on the existence of ELMMs (see Platen and Heath \cite{platen2006benchmark} for a detailed account).
Karatzas and Kardaras \cite{Karatzas2007} showed that 
the existence of a num\'{e}raire portfolio is equivalent to {\it No Unbounded Profit with Bounded Risk (NUPBR)},
which is weaker than the NFLVR condition.
The NUPBR condition has been found to be equivalent to {\it No Arbitrage of the First Kind (NA1)} (Kardaras \cite[Proposition 1]{Kardaras2010}),
and to the existence of equivalent supermartingale deflators (Kardaras \cite{Kardaras2012}).
{\it No Strong Arbitrage (NSA)}
and {\it No Increasing Profit (NIP)}
are weaker than the abovementioned conditions.
They excluded only arbitrage opportunities that yield a positive profit without requiring any initial investment or any amount of credit.
Generally, the NIP and NSA conditions are used only to exclude almost pathological notions of arbitrage,
while the NA1 condition is equivalent to an economically sound notion of market viability and helps resolve the fundamental problems of portfolio optimization, pricing, and hedging (see Fontana \cite[Section 8]{Fontana2015}).

Financial market models that violate the NIP condition have been proposed to consider problems in mathematical finance.
For example, Jarrow and Protter \cite{Jarrow2005} demonstrated that the existence of large traders, such as government or central banks, introduces reflection boundaries to asset prices in markets. 
Neuman and Schied \cite{NeumanSchied2016} studied optimal strategies in target zone models,
in which the price is modeled by a diffusion process with reflection boundaries,
such as a reflected arithmetic or geometric Brownian motion.
However, the theory of contingent claim valuation in the presence of increasing profit has not been well established.
Several authors, such as Veestraeten \cite{Veestraeten2008}, have proposed pricing formulae for a reflected geometric Brownian motion
under the mistaken belief that an ELMM exists.
Rossello \cite{Rossello2012} and Buckner et~al. \cite{Buckner2022} pointed out that a skew Brownian motion and a reflected geometric Brownian motion give rise to arbitrage opportunities, respectively. However,
they did not propose alternative pricing methodologies.
Jarrow and Protter \cite{Jarrow2005} discussed the completeness of the markets with increasing profits,
but their replicating strategies were not necessarily optimal.

The present study makes two contributions.

First, it develops a theory of contingent claim valuation
in financial market models that violate the NIP, NSA, and NA1 conditions.
In line with the literature,
this study defines
the price of a contingent claim as the smallest initial fortune to superreplicate the claim; that is, the {\it upper hedging price}.
The superreplication is possible under a predictable representation property even if arbitrage exists,
as Jarrow and Protter \cite{Jarrow2005} pointed out, assuming the uniqueness of ELMM.
Our ingredient is to find the smallest fortune for superreplication.
A typical case is as follows.
A trading strategy is carried out as in market models without arbitrage until arbitrage occurs, or otherwise, until maturity.
When arbitrage occurs, the fortune for superreplication thereafter can be raised using the arbitrage.
This allows the wealth process of the strategy to be zero at that time.
For example,
if the stock price process follows a reflected geometric Brownian motion,
the smallest fortune for replicating a call option is the Black--Scholes price
for the call option that knocks out at the reflection level.
Even with some exceptions,
the pricing formula is given as the expectation with respect to the physical probability measure
of the payoff multiplied by a kind of a weak martingale deflator
that vanishes when arbitrage occurs.
However, the presence of arbitrage does not necessarily reduce the initial fortune.

The second contribution of this study is that it models a stock price process incorporating effect of a large trader.
In the literature, for a market in which there is artificial price support by a large trader,
financial market models with singular terms have been used.
For example, Jarrow and Protter \cite{Jarrow2005} considered the stock price of a firm that was spun off from a large company
assuming that this large company offers a share purchase or sells at a fixed level.
The resultant price process has the local time at that level, which yields increasing profit.
The authors called these arbitrage opportunities ``hidden,'' because they occur during a small set of times of Lebesgue measure zero.
Our approach is different.
We assume that an agent in a corporate issues or repurchases its corporate stocks.
We model the number of the stocks,
put a constraint on wealth processes of market participants,
and derive a stock price process as a semimartingale with a singular term that comes from the number of the stocks.
Furthermore, we propose models for which increasing profit can be hidden from agents who observe the stock price process only.
More precisely, if the filtration generated by the stock price only
is strictly smaller than the original filtration,
increasing profit cannot be observed by agents whose information is represented with the smaller filtration.
A typical example is a case in which the stock price process follows the three-dimensional Bessel process.
This model satisfies the NA1 condition, which is stronger than the NIP condition, if the filtration is generated by the stock price.
However, according to Pitman's theorem,
the three-dimensional Bessel process is decomposed into a Brownian motion
and its future infimum process, which yields increasing profit.
Our interpretation is as follows:
an agent of a corporate repurchases the stocks,
which makes the stock price singularly increase;
however, this effect cannot be observed by market participants who observe the stock price only.
The proposed funding strategies can yield an arbitrage, as in this example,
and can avoid or lead to {\it bankruptcy}, 
which is defined as the event that the corporate stock is worthless.
We investigate these conditions with and without funding strategies.

In the following section, we introduce the notation and review the literature on arbitrage and contingent claim valuation.
In Section \ref{sec:ccv_under_arb}, we study contingent claim valuation in market models with increasing profit, strong arbitrage, and arbitrage of the first kind.
These forms of arbitrage are discussed in Sections \ref{sec:ip}, \ref{sec:sa}, \ref{sec:a1}, respectively.
Theorem \ref{thm:main} is the main theorem in this study, whose proof is in Section \ref{sec:proof}.
We model stock price processes incorporating funding activities in Section \ref{sec:funding}.
We demonstrate that the singular term appears in the semimartingale decomposition of the stock price process through the number of the stocks in Section \ref{sec:zero-sum}.
This singular term can be hidden from market participants who observe the stock price only in Section \ref{sec:fs_nip}.
Finally, Section \ref{sec:conclusion} concludes.

\section{Preliminaries}
\subsection{Financial market model}
Let $(\Omega,\mathcal{F},(\mathcal{F}_{t})_{t \ge 0},P)$ be a filtered probability space,
where the filtration $(\mathcal{F}_{t})_{t \ge 0}$ is assumed to satisfy the usual conditions of right-continuity and $P$-completeness.
We consider a financial market comprising $d+1$ assets with an infinite time horizon.
One of the assets is a savings account and we assume that its price is constant; that is, the interest rate is zero.
The other asset prices are described by the $\mathbb{R}_{+}^{d}$-valued process $S=(S_{t})_{t \ge 0}$ with $S_{t}=(S_{t}^{1},\dots,S_{t}^{d})^{\top}$, where $\mathbb{R}_{+}:=[0,\infty)$ and ${}^{\top}$ denotes transposition.
We denote by $(1,S)$ the financial market model comprising these $d+1$ assets.
Furthermore, the price process $S$ is assumed to be a continuous $\mathbb{R}_{+}^{d}$-valued semimartingale on $(\Omega,\mathcal{F},(\mathcal{F}_{t}),P)$ with semimartingale decomposition
\begin{equation}
S_{t} = S_{0} + M_{t} + A_{t},\label{eq:underlying}
\end{equation}
where $M$ is a $\mathbb{R}^{d}$-valued process with each component a local martingale
and
$A$ is a continuous $\mathbb{R}^{d}$-valued predictable process of finite variation with $M_{0}=A_{0}=0$.
Based on Jacod and Shiryaev \cite[Proposition II.2.9]{jacod2003},
a continuous increasing real-valued predictable process $\Lambda$ exists such that 
it holds that, for all $i,j=1,\dots,d$ and $t \ge 0$,
\begin{equation}
A_{t}^{i} = \int_{0}^{t} a_{u}^{i} d\Lambda_{u}, \;
\left<M^{i},M^{j}\right>_{t} = \int_{0}^{t} c_{u}^{ij} d\Lambda_{u},
\end{equation}
where $a=(a^{1},\dots,a^{d})^{\top}$ is an $\mathbb{R}^{d}$-valued predictable process
and $c = ((c^{i1})_{1 \le i \le d},\dots,(c^{id})_{1 \le i \le d})$ is a predictable process taking values in the cone of symmetric nonnegative $d \times d$-matrices.
Let us denote by $c_{t}^{+}$ for each $t \ge 0$ the Moore--Penrose pseudoinverse of the matrix $c_{t}$.
Then, the process $a$ can be represented as 
\begin{equation}
a = c \lambda + \nu,
\end{equation}
where $\lambda = (\lambda_{t})_{t \ge 0}$ with $\lambda_{t} := c_{t}^{+} a_{t}$
and $\nu= (\nu_{t})_{t \ge 0}$ are $\mathbb{R}^{d}$-valued predictable processes with $\nu_{t} \in \mathrm{Ker}(c_{t}):=\{x \in \mathbb{R}^{d} : c_{t} x = 0 \}$ for each $t \ge 0$.
We further decompose $A$ into 
\begin{equation}
A_{t} = \int_{0}^{t} d\left<M,M\right>_{u} \lambda_{u} + \kappa_{t},\;
\kappa_{t} := \int_{0}^{t} \nu_{u} d\Lambda_{u}.
\end{equation}

We suppose that the financial market is frictionless in the sense that there are no trading restrictions, transaction costs, and so on (we assume the existence of large traders that affect the asset prices in Section \ref{sec:funding}).
The activity of trading is represented by the stochastic integral with respect to $S$.
Let $L(S)$ be the set of all $\mathbb{R}^{d}$-valued $S$-integrable predictable processes, in the sense of Jacod and Shiryaev \cite[III.6]{jacod2003}.
For $h =(h^{i})_{i=1,\dots,d} \in L(S)$, the stochastic integral process $(\int_{0}^{t} h_{u} dS_{u})_{t \ge 0}$ is denoted by $h \cdot S$
and the process $V(x,h) := (V_{t}(x,h))_{t \ge 0}$ for $x \in \mathbb{R}$ is defined as $V_{t}(x,h) : = x+(h \cdot S)_{t}$.
Then, the process $V(x,h)$ expresses the wealth process of the self-financing strategy
such that the initial fortune is $x$
and the number of units of the $i$-th stock held at time $t$ is $h_{t}^{i}$.
We introduce the notion of {\it admissible strategy} in order to ban so-called doubling strategies as follows:
\begin{definition}
For $a \in [0,\infty)$,
an element $h \in L(S)$ is said to be an $a$-admissible strategy
if $P[(h \cdot S)_{t} \ge -a \text{ for all } t \in [0,\infty)]=1$ holds.
An element $h \in L(S)$ is said to be an admissible strategy
if it is an $a$-admissible strategy for some $a \ge 0$.
The set of all $a$-admissible strategies are denoted by $\mathcal{A}_{a}$,
and $\mathcal{A} := \bigcup_{a \ge 0} \mathcal{A}_{a}$.
\end{definition}

The {\it num\'{e}raire portfolio} and {\it martingale deflators} play important roles in mathematical finance.
The existence of the num\'{e}raire portfolio is the minimal a-priori assumption required to proceed with utility optimization (Karatzas and Kardaras \cite{Karatzas2007}).
It is important for contingent claim valuation in the Benchmark approach (Platen and Heath \cite{platen2006benchmark}).
The num\'{e}raire portfolio is also called growth optimal and benchmark portfolio (see Hulley and Schweizer \cite[Theorem 2.1]{Hulley2010}).
\begin{definition}
The portfolio generated by the strategy $\widehat{h} \in \mathcal{A}_{1}$ is called num\'{e}raire portfolio,
if the relative wealth process $V(1,h)/V(1,\widehat{h})$
is a supermartingale for every $h \in \mathcal{A}_{1}$.
\end{definition}
\begin{definition}
A nonnegative local martingale $Z=(Z_{t})_{t \ge 0}$ with $Z_{0}=1$
is called a weak martingale deflator if the product $ZS^{i}$ is a local martingale for all $i=1,\dots,d$
and $Z$ is called a martingale deflator if $Z$ satisfies $Z_{t}>0$ for $t \ge 0$ in addition.
\end{definition}

\subsection{Notions of arbitrage}
Several forms of arbitrage have been proposed in the literature.
Among them, we review the definitions of {\it increasing profit}, {\it strong arbitrage}, {\it arbitrage of the first kind}, {\it arbitrage}, and {\it Free Lunch with Vanishing Risk} below:
\begin{definition}
Let $T \in (0,\infty)$ be a constant.
\begin{itemize}
\item[(i)] An element $h \in \mathcal{A}_{0}$ generates an {\rm increasing profit}
if $P[(h\cdot S)_{T} >0] > 0$
and $P[(h\cdot S)_{s} \le (h\cdot S)_{t}\text{ for all } 0 \le s \le t \le T]=1$;
\item[(ii)] an element $h \in \mathcal{A}_{0}$ generates a {\rm strong arbitrage opportunity}
if $P[(h\cdot S)_{T}>0]>0$;
\item[(iii)] a nonnegative random variable $\xi$ generates an {\rm arbitrage of the first kind}
if $P[\xi >0]>0$ and for every $\varepsilon \in (0,\infty)$,
there exists an element $h^{\varepsilon} \in \mathcal{A}_{\varepsilon}$
such that $P[V_{T}(\delta,h^{\varepsilon}) \ge \xi]=1$;
\item[(iv)] an element $h \in \mathcal{A}$ generates an {\rm arbitrage opportunity} 
if $P[(h\cdot S)_{T}>0]>0$ and $P[(h\cdot S)_{T} \ge 0]=1$;
\item[(v)] a sequence $\{h^{n}\}_{n \in \mathbb{N}}$ generates a {\rm Free Lunch with Vanishing Risk},
if there exists an $\varepsilon > 0$
and an increasing sequence $\{\eta_{n}\}_{n \in \mathbb{N}}$ with $0 \le \eta_{n} \nearrow 1$
such that $P[(h^{n}\cdot S)_{T}>\varepsilon] > \varepsilon$
and $P[(h^{n}\cdot S)_{T}>-1+\eta_{n}]=1$, for all $n \in \mathbb{N}$.
\end{itemize}
The acronyms NIP, NSA, NA1, NA, and NFLVR for the corresponding types of no-arbitrage are used.
These no-arbitrage conditions, if the time horizon $T$ is omitted, are understood for all $T \in (0,\infty)$.
\end{definition}

The implications about these no-arbitrage conditions are as follows:
\begin{equation}
\mathrm{NFLVR} \Longleftrightarrow
\mathrm{NA \& NA1} \Longrightarrow
\mathrm{NA1} \Longrightarrow
\mathrm{NSA} \Longrightarrow
\mathrm{NIP}. \label{eq:implication}
\end{equation}
The none of the converse implications holds in general.
For the proof and explicit examples and counterexamples, see Fontana \cite{Fontana2015}, 
which provided a unified analysis of no-arbitrage conditions
for financial market models based on continuous semimartingales
putting an emphasis on the role of martingale deflators.

An increasing profit is the strongest form of arbitrage for continuous-time financial market models
and the NIP condition has been considered to be banned from any reasonable market model.
However, financial market models that violate the NIP condition are considered for markets under government or central bank interventions, 
such as interest rates or exchange rate markets.
The NIP condition is equivalent to $\nu_{\cdot} = 0$ holds $P \otimes \Lambda$-a.e.,
or equivalently, $P[\rho^{*} = \infty]=1$ (Fontana \cite[Theorem 3.1]{Fontana2015}), where
\begin{equation}
\rho^{*} := \inf \{t \ge 0 : |\kappa|_{t} > 0 \},\;
|\kappa|_{t} := \int_{0}^{t} |\nu_{u}|d\Lambda_{u}
\end{equation}
with the usual convention $\inf \emptyset = \infty$.

The NSA condition is also regarded as the necessary condition for any reasonable market model.
Provided that the NIP condition holds,
the NSA condition is equivalent to $P[\rho^{+} = \infty]=1$ (Fontana \cite[Theorem 4.3]{Fontana2015}),
where 
\begin{equation}
\rho^{+} := \inf \{t > 0 : \widehat{K}_{t}^{t+\delta} = \infty \text{ for all } \delta > 0\},
\end{equation}
and 
\begin{equation}
\widehat{K}_{s}^{t} := \int_{s}^{t} \lambda_{u}^{\top} d\left<M,M\right>_{u} \lambda_{u}\label{eq:mean-variance}
\end{equation}
for $s,t \ge 0$ with $s < t$
and $\widehat{K}_{s}^{t} :=0$ with $t \le s$
(we always define $\int_{s}^{t}:=0$ with $t \le s$ for any kind of integrals).
The process $(\widehat{K}_{t})_{t\ge0} := (\widehat{K}_{0}^{t})_{t\ge 0}$ is called the {\it mean-variance trade-off process}.
The NSA condition turns out to be equivalent to the absence of {\it immediate arbitrage},
which was first formulated in Delbaen and Schachermayer \cite{DelbaenSchachermayer1995b} (Fontana \cite[Lemma 4.2]{Fontana2015}).

In contrast to the above two no-arbitrage conditions,
the NA1 condition is considered to be the minimal assumption for a reasonable financial market model.
More precisely, Karatzas and Kardaras \cite{Karatzas2007} 
showed NUPBR as the minimal a-priori assumption required to proceed with utility optimization,
and this condition was proved to be equivalent to the NA1 condition (Fontana \cite[Lemma 5.2]{Fontana2015}).
Provided that the NSA condition holds,
the NA1 condition holds
if and only if $P[\rho^{1} = \infty]=1$ (Fontana \cite[Theorem 4.3]{Fontana2015}),
where
\begin{equation}
\rho^{1} := \inf \{t >0 : \widehat{K}_{t} = \infty \}.
\end{equation}
We remark that $P[\rho^{1} \le \rho^{+}]=1$.

The notions of NFLVR and NA are classical.
Delbaen and Schachermayer \cite{Delbaen1994} proved, in the case of locally bounded process (this assumption is removed later),
the equivalence between the NFLVR condition
and
the existence of an {\it Equivalent Local Martingale Measure}; that is,
a probability measure equivalent to the original one such that the discounted asset price process is a local martingale under the new measure.

We define a right-continuous nonnegative supermartingale
$\widehat{Z}=(\widehat{Z}_{t})_{t\ge 0}$ as
\begin{equation}
\widehat{Z}_{t}
:= \exp\left(-\int_{0}^{t}\lambda_{u} dM_{u}-\frac{1}{2}\int_{0}^{t} \lambda_{u}^{\top} d\left<M,M\right>_{u} \lambda_{u} \right)1_{\{t < \rho^{1}\}}. \label{eq:defrator}
\end{equation}
We remark that
$\widehat{Z}_{\rho^{1}} = 0$ for $\rho^{1} < \infty$
and $\widehat{Z}_{\cdot}(\omega)$ is not continuous at $\rho^{1}(\omega)$ for
$\omega \in \Omega$
such that $\widehat{K}_{\cdot}(\omega)$ jumps to $\infty$ at $\rho^{1}(\omega)<\infty$,
or equivalently, $\rho^{1}(\omega)<\infty$
and $\widehat{K}_{\rho^{1}}(\omega)<\infty$.
Under the NA1 condition,
$\widehat{Z}$ is a continuous martingale deflator,
which is called the {\it minimal martingale deflator} (Schweizer \cite{Schweizer1995}),
and
its reciprocal $\widehat{G} := 1/\widehat{Z}$ is a num\'{e}raire portfolio,
for which $h := \widehat{G} \lambda \in \mathcal{A}_{1}$
is the trading strategy; that is, $\widehat{G} = V(1,h)$.
The $\widehat{G}$-denominated market model $(1/\widehat{G},S/\widehat{G})$ of the original model $(1,S)$ satisfies the NFLVR condition,
where the original probability $P$ is itself an ELMM.
However, in the presence of stronger arbitrage,
$\widehat{G}$ is not well defined as a real-valued process after $\rho^{1}$,
although $\widehat{Z}$ is well-defined.
Instead, we introduce $\widehat{L}=(\widehat{L}_{t})_{t \ge 0}$ as $\widehat{L}_{0}=1$ and $\widehat{L}_{t} :=\widehat{Z}_{t-}$ for $t > 0$.
We remark that $\widehat{L}_{\cdot}(\omega)$ and $\widehat{Z}_{\cdot}(\omega)$ may differ at $\rho^{1}(\omega) < \infty$;
that is, $\widehat{L}_{\rho^{1}}(\omega) > \widehat{Z}_{\rho^{1}}(\omega) = 0$
for $\omega \in \Omega$ such that $\rho^{1}(\omega) < \infty$
and $\widehat{K}_{\rho^{1}}(\omega)<\infty$.
Instead of $\widehat{Z}$ and $\widehat{G}$, we use $\widehat{L}$ and $1/\widehat{L}$ appropriately stopped.
\begin{lemma}
\label{lem:L}
Let $\theta_{\varepsilon}^{1} := \inf \{t >0 : \varepsilon > \widehat{L}_{t} \}$
for $\varepsilon \in (0,1)$.
Then, the stopped process $\widehat{L}^{\rho^{1}}$ is a nonnegative continuous local martingale that satisfies
\begin{equation}
\widehat{L}_{t}^{\theta_{\varepsilon}^{1}}
= 1 - \int_{0}^{t \wedge \theta_{\varepsilon}^{1}} \widehat{L}_{t}\lambda_{u} dM_{u},
\end{equation}
and the reciprocal $1/\widehat{L}^{\theta_{\varepsilon}^{1}}$
is a continuous semimartingale that
is bounded by $1/\varepsilon$
and is
expressed with a stochastic integral with respect to $S$:
\begin{equation}
\frac{1}{\widehat{L}_{t}^{\theta_{\varepsilon}^{1}}} = 1 +
\int_{0}^{t \wedge \theta_{\varepsilon}^{1}}\frac{1}{\widehat{L}_{u}} \lambda_{u} dS_{u}. 
\end{equation}
\end{lemma}
\begin{proof}
The left-continuity of $\widehat{L}$ yields $\widehat{L}_{\theta_{\varepsilon}^{1}} \ge \varepsilon$.
The integral expressions of $\widehat{L}^{\theta_{\varepsilon}^{1}}$ and $1/\widehat{L}^{\theta_{\varepsilon}^{1}}$ are straightforward.
\end{proof}

\subsection{Contingent claim valuation}
\label{sec:Contingent claim valuation}
Contingent claim valuation is discussed in a market model under a no-arbitrage condition with an appropriate class of admissible strategies.
We consider the price of the contingent claim that pays $\xi$ at time $T \in (0,\infty)$ 
for a nonnegative $\mathcal{F}_{T}$-measurable random variable $\xi$.
A standard definition is the so-called {\it upper hedging price},
which is defined as the smallest initial fortune required to finance an admissible superreplicating strategy:
\begin{equation}
\mathcal{U}(\xi)
:= \inf \{x \in \mathbb{R} \;|\text{ there exists a strategy $h \in \mathcal{A}$ with $\xi \le V_{T}(x,h)$} \}.
\end{equation}

Under the NFLVR condition, the price for the contingent claim can be calculated as the following duality equality
\begin{equation}
\sup_{Q} E^{Q}[\xi] = \mathcal{U}(\xi),\label{eq:price-duality}
\end{equation}
provided that $\xi$ is bounded,
where $Q$ in the $\sup$ of the right-hand side is an ELMM and $E^{Q}$ is the expectation operator with respect to $Q$
(Delbaen and Schachermayer \cite[Theorem 5.7]{DelbaenSchachermayer1995b}).
More generally, \eqref{eq:price-duality} holds with $\le$ instead of the equality.

Without ELMMs, contingent claim valuation is still possible
in a market model that satisfies the NA1 condition,
where the num\'{e}raire portfolio $\widehat{G}$ is well defined.
Because the $\widehat{G}$-denominated financial market model $(1/\widehat{G},S/\widehat{G})$ satisfies the NFLVR condition, 
the contingent claim valuation in this model is reduced to the abovementioned approach.
The upper hedging price in the $\widehat{G}$-denominated financial market model is greater than
\begin{equation}
E\left[\frac{\xi}{\widehat{G}_{T}}\right] = E[\xi \widehat{Z}_{T}],
\end{equation}
because $P$ is an ELMM for the market model $(1/\widehat{G},S/\widehat{G})$.
This is the price by the Benchmark approach (Platen and Heath \cite{platen2006benchmark}), and its generalization of Ruf \cite{ruf2012}.
We stress that the notions of $\mathcal{A}$ and $\mathcal{A}_{a}, a > 0$ are dependent on the num\'{e}raire
and
that the pair $(x,h) \in \mathbb{R} \times \mathcal{A}$ of admissible strategies in the Benchmark approach are restricted to
those with $V_{t}(x,h) \ge 0$ for $t \in [0,T]$.
Carr et~al. \cite{CarrFisherRuf2014} proposed a pricing operator based on 
superreplicating strategies that succeed with probability one both for the original market model and $\widehat{G}$-denominated one.
See Fontana and Runggaldier \cite[Section 5]{FontanaRunggaldier2013} for other approaches under the NA1 condition.

Contingent claim valuation in market models that violate the NA1 condition has rarely been the subject of investigation, to the best of our knowledge.
A notable exception was Jarrow and Protter \cite{Jarrow2005}.
They considered valuation and replication
in market models with increasing profit and with strong arbitrage (see Jarrow and Protter \cite[Sections 4 and 7]{Jarrow2005}
and Examples \ref{ex:rgbm} and \ref{ex:jarrowprotter} of this study).
Rossello \cite{Rossello2012} pointed out that a skew Brownian motion gives rise to some kind of arbitrage
and Buckner et~al. \cite{Buckner2022} demonstrated
that the option pricing formulae based on a reflected Brownian motion proposed by Veestraeten \cite{Veestraeten2008} violate classical no-arbitrage bounds.
They, however, did not propose any alternative approach.

\subsection{Examples based on the three-dimensional Bessel process}
\label{sec:eg_bessel}
Based on the three-dimensional Bessel process,
we construct several financial market models with several forms of arbitrage,
where a savings account with zero interest rate and a stock are traded.
We begin with a Wiener measure
and then obtain a three-dimensional Bessel process by an absolutely continuous change of probability measures.
Then, we develop market models by enlarging filtrations and conditioning.

We work on the canonical sample space $(\mathscr{C},\mathscr{B}(\mathscr{C}))$,
where 
$\mathscr{C} :=C([0,\infty))$ is the space of all continuous real-valued functions on $[0,\infty)$
and $\mathscr{B}(\mathscr{C})$ is its Borel $\sigma$-filed.
Let $X$ denote the coordinate process; that is, $X_{t}(\omega):=\omega(t)$ for $\omega\in \mathscr{C}$,
$(\mathcal{G}_{t}^{0})$ be the filtration generated by $X$,
and $\tau_{j} := \inf \{t > 0: X_{t} = j\}$ be the first hitting time to $j \in \mathbb{R}$.
On this sample space,
we define the Wiener measure $W_{x}$ for $x \in \mathbb{R}$; 
that is, the law of the Brownian motion started at $x$,
and a probability measure $P_{x}^{(1/2)}$ by
\begin{equation}
P_{x}^{(1/2)} = \frac{X_{t \wedge \tau_{0}}}{x} \cdot W_{x} \quad \text{ on } \quad \mathcal{G}_{t}^{0}, \label{eq:conditioning}
\end{equation}
which is the law of the three-dimensional Bessel process started at $x > 0$.
Similarly, we define the laws $P_{x}^{(\nu)}$ of the Bessel process with index $\nu$ (see Revuz and Yor \cite[Exercise XI.1.22]{revuzyor}, or Borodin and Salminen \cite[page 138]{BorodinSalminen2002}):
\begin{equation}
P_{x}^{(\nu)}
= \left(\frac{X_{t\wedge \tau_{0}}}{x}\right)^{\nu+1/2}
\exp\left(-\frac{\nu^{2}-1/4}{2}\int_{0}^{t}\frac{du}{X_{u}^{2}}\right) \cdot W_{x}  \quad \text{ on } \quad \mathcal{G}_{t}^{0}. \label{eq:bessel_ac}
\end{equation}
Furthermore, the laws $P_{x}^{(\nu)}$ can be extended continuously to $P_{0}^{(\nu)}$ for $\nu > 0$,
although they are not absolutely continuous to $W_{0}$.
For these probability measures $W_{x}$ and $P_{x}^{(\nu)}$, the expectation operator is also denoted by $W_{x}$ and $P_{x}^{(\nu)}$, respectively.

\subsubsection{Bachelier's model}
\label{sec:bachelier}
We begin with the stock price process $X$ with the probability measure $W_{x}$ for $x > 0$ 
and the usual augmentation of the filtration generated by $X$.
This is Bachelier's model (Bachelier \cite{Bachelier1900}).
The market model $(1,X)$ satisfies the NFLVR condition, where $W_{x}$ is an ELMM.
We abused the notation because we defined a stock price as a nonnegative semimartingale.

In this model, the price at time $0$ of the contingent claim that pays $g(X_{\tau})1_{\{\tau_{j} > \tau\}}$ at time $\tau$ is
\begin{equation}
\Psi^{g}(\tau,x,j)
:=  W_{x} \left[g(X_{\tau}) : \tau_{j} > \tau \right],
\end{equation}
where $g :[0,\infty) \longrightarrow [0,\infty)$ is a Borel measurable function
with $g(X_{\tau})1_{\{\tau_{j} > \tau\}} \in L^{1}(W_{x})$.
We define $\Psi_{x}^{g}(\cdot,x,j)$ and $\Psi_{xx}^{g}(\cdot,x,j)$
as the first- and second- derivative for $x \neq j$, respectively, 
and $\Psi_{x}^{g}(\cdot,j,j) := \frac{1}{2}(\Psi_{x+}^{g}(\cdot,j,j)+\Psi_{x-}^{g}(\cdot,j,j))$
and $\Psi_{xx}^{g}(\cdot,j,j) :=\Psi_{x+}^{g}(\cdot,j,j)-\Psi_{x-}^{g}(\cdot,j,j)$, where
$\Psi_{x-}^{g}$ and $\Psi_{x+}^{g}$ are the left- and right- limit of $\Psi_{x}^{g}$, respectively.
If $\Psi^{g}$ is sufficiently smooth with $\tau$ and $x$,
$\Psi^{g}$ satisfies
\begin{equation}
\Psi_{\tau}^{g} = \frac{1}{2}\Psi_{xx}^{g} \label{eq:heat}
\end{equation}
in an appropriate domain.

\subsubsection{Three-dimensional Bessel process}
By an absolutely continuous change of probability measures \eqref{eq:conditioning},
we obtain the stock price process $X$ on the probability measure
$(\mathscr{C},\mathcal{G}_{\infty},(\mathcal{G}_{t}),P_{x}^{(1/2)})$
with $x \ge 0$,
where $(\mathcal{G}_{t})$ be the $P_{x}^{(1/2)}$-augmented natural filtration of $X$
and $\mathcal{G}_{\infty} = \vee_{t \ge 0} \mathcal{G}_{t}$.
This is one of the best-known examples for market models for which the NFLVR condition fail (see Delbaen and Schachermayer \cite{DelbaenSchachermayer1995Bessel}).

If $x=0$, the NSA condition is violated,
where a strong arbitrage occurs at $\rho^{+} = 0$.

The model for $x > 0$ does not satisfy the NA condition, but the NA1 condition, where $\widehat{Z}=x/X$ is a martingale deflator
and $\widehat{G} = X/x$ is the num\'{e}raire portfolio.
An arbitrage strategy is explicitly demonstrated.
For $\tau>0$ and $x>0$, let 
\begin{equation}
\Psi(\tau,x) := P_{x}^{(1/2)}\left[\frac{X_{0}}{X_{\tau}}\right]= W_{x}[\tau < \tau_{0}].
\end{equation}
Then, for $T>0$, the wealth process of the strategy with the initial fortune $0$
and $\Psi_{x}(T-t,X_{t})$-units of the stock at time $t \le T$
is
\begin{equation}
\int_{0}^{t} \Psi_{x}(T-u,X_{u}) dX_{u} = \Psi(T-t,X_{t})-\Psi(T,x)
\ge -\Psi(T,x).
\end{equation}
This shows that the strategy is admissible
and yields $1-\Psi(T,x)>0$ at time $T$ with initial fortune $0$.
However, this strategy is not admissible in the $\widehat{G}$-denominated market model $(x/X,x)$,
because the $\widehat{G}$-denominated wealth process $(\Psi(T-t,X_{t})-\Psi(T,x))x/X_{t}$
is not bounded below.
No initial fortune is required for replicating the contingent claim that pays $1$ at time $T$ in the market model $(1,X)$,
while the price by the Benchmark approach is
\begin{equation}
P_{x}^{(1/2)}\left[\frac{X_{0}}{X_{T}}\right]=\Psi(T,x).
\end{equation}
The savings account denominated by the stock price $X$ 
is a strict local martingale (i.e., a local martingale but not a true martingale) with respect to $P_{x}^{(1/2)}$.
This kind of market models is called a bubble model
and has been studied in the literature (Cox and Hobson \cite{CoxHobson2005}).

\subsubsection{Enlargement with the future infimum process $J$}
Let $J$ be the future infimum process of $X$; that is, $J_{t} := \inf_{u>t} X_{u}$.
Then, Pitman's theorem (Revuz and Yor \cite[Theorem VI.3.6]{revuzyor}) shows that
a $P_{x}^{(1/2)}$-Brownian motion $\beta$ exists such that 
\begin{equation}
X_{t} = \beta_{t}+2J_{t}.
\end{equation}
The financial market model $(1,X)$ on the probability measure
$(\mathscr{C},\mathcal{G}_{\infty},(\mathcal{G}_{t}^{J}),P_{x}^{(1/2)})$ with $x>0$
violates the NIP condition,
where $\mathcal{G}_{t}^{J} := \mathcal{G}_{t} \vee \sigma (J_{t})$.

\subsubsection{Enlargement with a future infimum process $J_{0}$}
\label{sec:williams}
Let $(\mathcal{G}_{t}^{J_{0}})$ be the initial enlargement of 
the filtration $(\mathcal{G}_{t})$ with $J_{0}$
and $\tau_{J_{0}}:=\inf \{t > 0: X_{t} = J_{0}\}$.
Then, Williams's path decomposition (Revuz and Yor \cite[Theorem VI.3.11]{revuzyor}) shows 
that $X$ satisfies 
\begin{equation}
X_{t} = X_{0} + \beta_{t \wedge \tau_{J_{0}}} + R_{(t-\tau_{J_{0}})_{+}},
\end{equation}
where $\beta$ is a Brownian motion with $\beta_{0}=0$
and $R$ is a three-dimensional Bessel process with $R_{0}=0$.
Then, the financial market model $(1,X)$ on the probability measure
$(\mathscr{C},\mathcal{G}_{\infty},(\mathcal{G}_{t}^{J_{0}}),P_{x}^{(1/2)})$
with $x>0$
violates the NSA condition: $\rho^{+}=\tau_{J_{0}}$.

\subsubsection{Enlargement with an honest time $\tau_{J_{0}}$}
The $(\mathcal{G}_{t}^{J})$-stopping time $\tau_{J_{0}}$ 
is not a stopping time with respect to $(\mathcal{G}_{t})$,
but an honest time: $\tau_{J_{0}}= \sup\{t : 1/X_{t} = \max_{u>0} 1/X_{u}\}$.
Let $(\mathcal{G}_{t}^{\tau_{J_{0}}})$ be the progressive enlargement of 
the filtration $(\mathcal{G}_{t})$ with $\tau_{J_{0}}$.
Then, the financial market model $(1/X^{\tau_{J_{0}}},1)$ on the probability measure
$(\mathscr{C},\mathcal{G}_{\infty},(\mathcal{G}_{t}^{\tau_{J_{0}}}),P_{x}^{(1/2)})$ with $x>0$
violates the NA condition.
For example, buying the saving account at time $0$ and selling at time $\tau$
yields a profit $1/X_{\tau_{J_{0}}}-1/X_{0} > 0$.

This model is very different from the others in the sense that the predictable representation property, Assumption \ref{ass:prp} below, is not satisfied.
Fontana et~al. \cite{Fontana2014} studied whether the additional information associated with an honest time gives rise to arbitrage profits.
We do not treat this model in our study and leave the investigation to future research.

\section{Contingent claim valuation under arbitrage}
\label{sec:ccv_under_arb}
We consider the price of the contingent claim that pays $\xi$ at a finite stopping time $\zeta$ 
for a nonnegative $\mathcal{F}_{\zeta}$-measurable random variable $\xi$.
From the point of view of the seller of the contingent claim,
the payoff $\xi$ is a liability that has to be covered with the right amount of  fortune at time $t=0$ and the right trading strategy.
The price of the contingent claim in this study is defined as the smallest
initial fortune required to finance an admissible superreplicating strategy.
\begin{definition}
\label{def:price}
For a nonnegative $\mathcal{F}_{\zeta}$-measurable random variable $\xi$,
the price of the contingent claim that pays $\xi$ at time $\zeta$ is defined as 
\begin{equation}
\inf \{x \ge 0 : h \in \mathcal{A}_{x} {\text{ such that }} \xi \le V_{\zeta}(x,h) \}.
\end{equation}
\end{definition}
\iffalse
\begin{remark}
\label{rem:creditline}
In Definition \ref{def:price}, candidate strategies $(x,h)$ are restricted to those with $V_{t}(x,h) \ge 0$ for $t \in [0,\zeta]$.
We could have relaxed this condition to those with $V_{t}(x,h) \ge -c$ for some $c > 0$.
Then, the fair price could be arbitrarily small,
which is demonstrated in Example \ref{ex:rgbm} below.
\end{remark}
\fi

We work under the following assumptions in this section.
\begin{assumption}
\label{ass:trivial}
The initial $\sigma$-field $\mathcal{F}_{0}$ is trivial.
\end{assumption}
\begin{assumption}
\label{ass:prp}
The continuous local martingale $M$ has the predictable representation property with respect to $(\mathcal{F}_{t})$;
that is,
any $(\mathcal{F}_{t})$-local martingale $N=(N_{t})$ with $N_{0}=0$
can be represented as $N=h \cdot M$ for some $h \in L(M)$.
\end{assumption}
\begin{remark}
\label{rem:prp}
Assumption \ref{ass:prp} ensures that all local martingale have continuous versions.
In particular, we always can and do take a continuous version for a martingale in a form $E[\cdot | \mathcal{F}_{t}]$.
\end{remark}

The following is the main theorem in this study.
\begin{theorem}
\label{thm:main}
Suppose that Assumptions \ref{ass:trivial} and \ref{ass:prp} hold.
Then, the price of the contingent claim that pays a bounded $\mathcal{F}_{\zeta}$ random variable $\xi$ at time $\zeta$ is
\begin{equation}
\xi_{0} :=
E\left[\xi \widehat{L}_{\zeta}1_{\{\zeta \le \rho\}} \right]
+
E\left[\xi_{\rho^{+}} \widehat{L}_{\rho^{+}} 1_{\{\rho^{+} = \rho < \zeta\}}\right], \label{eq:main}
\end{equation}
where $\rho := \rho^{*} \wedge \rho^{+} \wedge \rho^{1}$
and 
$\xi_{\rho^{+}}$ is an $\mathcal{F}_{\rho^{+}}$-measurable random variable
and is the price at time $\rho^{+} < \zeta$.
\end{theorem}

We treat contingent claim valuation focusing on each form of arbitrage in the following three subsections
and prove Theorem \ref{thm:main} in the final subsection.
We use neither Assumptions \ref{ass:trivial} and \ref{ass:prp}
nor the boundedness of $\xi$ to prove that the price is greater than $\xi_{0}$.
\begin{proposition}
\label{prop:minimality}
Let $\zeta$ be a finite stopping time,
$\xi$ be a $\mathcal{F}_{\zeta}$-measurable random variable,
and $\xi_{\rho^{+}}$ be a $\mathcal{F}_{\rho^{+}}$-measurable random variable,
which is the price at time $\rho^{+} < \zeta$
of the contingent claim that pays $\xi$ at time $\zeta$.
Suppose that $\xi \widehat{L}_{\zeta}1_{\{\zeta \le \rho\}}$ and $\xi_{\rho^{+}} \widehat{L}_{\rho^{+}}1_{\{\rho = \rho^{+} < \zeta\}}$
are integrable.
Then, for $a \ge 0$ and $h \in \mathcal{A}_{a}$ such that $\xi \le V_{\zeta}(a,h)$,
we obtain $\xi_{0} \le a$.
\end{proposition}
\begin{proof}
Let $\widehat{N}_{t} := E[\xi \widehat{L}_{\zeta} 1_{\{\zeta \le \rho\}}+\xi_{\rho^{+}}\widehat{L}_{\rho^{+}} 1_{\{\rho = \rho^{+} < \zeta\}} | \mathcal{F}_{t}]$
and $N_{t} := (a+(h\cdot S)_{t}) \widehat{L}_{t}^{\rho^{+}}$ for $t\ge 0$.
Then, $\widehat{N}$ is a martingale
and $N$ is a nonnegative continuous local martingale.
The minimality of $\xi_{\rho^{+}}$ yields $\xi_{\rho^{+}} \le a+(h\cdot S)_{\rho^{+}}$ on $\{\rho^{+} < \zeta\}$. 
Hence, we obtain
\begin{equation}
\widehat{N}_{\rho \wedge \zeta}
= \xi \widehat{L}_{\zeta} 1_{\{\zeta \le \rho\}}+\xi_{\rho^{+}}\widehat{L}_{\rho^{+}} 1_{\{\rho = \rho^{+} < \zeta\}}
\le N_{\zeta} 1_{\{\zeta \le \rho\}}+N_{\rho}1_{\{\rho = \rho^{+} < \zeta\}}
\le N_{\zeta \wedge \rho}
\end{equation}
by the nonnegative of $N$.
The conclusion follows from $\xi_{0} = \widehat{N}_{0} = E[\widehat{N}_{\zeta \wedge \rho} | \mathcal{F}_{0}] \le E[N_{\zeta \wedge \rho} | \mathcal{F}_{0}] \le a$.
\end{proof}

\subsection{Contingent claim valuation under increasing profits}
\label{sec:ip}
In this subsection, we consider the price of the contingent claim that pays $\xi$ at a finite stopping time $\zeta$ under increasing profit
with the condition $\widehat{K}_{0}^{\zeta} < \infty$.

We remark that
$c_{t} c_{t}^{+}$ is the orthogonal projection operator onto the linear subspace $\mathrm{Im}(c_{t})=\mathrm{Ker}(c_{t})^{\bot}$.
In addition,
$(c_{t}^{ij})_{1 \le i \le d} \in \mathrm{Im}(c_{t})$ for $j = 1,\dots,d$
and
$\nu_{t} \in \mathrm{Ker}(c_{t})$
yield $(h-c^{+} c h) \cdot \left<M,M\right>=0$
and $(c c^{+} h) \cdot \kappa=0$, respectively, for $h \in L(S)$.
Hence, we obtain
\begin{equation}
(c c^{+} h) \cdot S
= (c c^{+} h) \cdot (M+  \lambda \cdot \left<M,M\right>)
  = h \cdot (M+ \lambda \cdot \left<M,M\right>)
  = h \cdot (S-\kappa). \label{eq:orthogonalIm}
\end{equation}

Although the stochastic exponential $\widehat{Z}$ and its reciprocal
$\widehat{G}$ are well-defined up to $\zeta$
under the condition $\widehat{K}_{0}^{\zeta} < \infty$,
we cannot use them as the minimal martingale deflator
and the num\'{e}raire portfolio, respectively,
in the presence of increasing profit.
Instead, we introduce an imitation of the num\'{e}raire portfolio
$1/L_{t}^{n} \in (0,\infty]$ for $n \ge 0$,
where $L^{n}:=(L_{t}^{n})_{t \ge 0}$ is defined as
\begin{equation}
L_{t}^{n} := \widehat{L}_{t}^{\rho^{1}}\mathrm{e}^{-n |\kappa|_{t\wedge \rho^{1}}}
\end{equation}
for $t \in [0,\infty)$.
While neither $L^{n}$ nor the product $S^{i}L^{n},i=1,\dots,d,$ are local martingales,
$1/L^{n}$ is tradable up to $\zeta$:
\begin{equation}
\frac{1}{L_{t \wedge \zeta}^{n}} = 1 + \int_{0}^{t \wedge \zeta} \frac{1}{\widehat{L}_{u}^{n}} \lambda_{u}^{n} dS_{u},
\end{equation}
with $\lambda_{t}^{n}:=(c_{t}c_{t}^{+}\lambda_{t}+(n /|\nu_{t}|)\nu_{t})1_{\{t \le \zeta\}}$.
More generally, 
the multiplication by $1/L^{n}$ transforms a stochastic integral
with respect to $M$ to that with respect to $S$.
\begin{lemma}
\label{lem:tradable}
Suppose that $\widehat{K}_{0}^{\zeta} < \infty$.
For $h \in L(M)$,
there exists $h^{n} \in L(S)$ such that $(h \cdot M)_{t \wedge \zeta}/L_{t \wedge \zeta}^{n} = (h^{n} \cdot S)_{t \wedge \zeta}$
for $t \ge 0$.
\end{lemma}
\begin{proof}
A straightforward computation leads to 
\begin{align}
(h \cdot M)/L^{n}
&= \int_{0}^{\cdot}(h \cdot M)_{u} d(1/L_{u}^{n}) + \int_{0}^{\cdot} 1/L_{u}^{n}h_{u} dM_{u} + \int_{0}^{\cdot} h_{u}^{\top} d\left<M,1/L^{n}\right>_{u} \nonumber\\
&= \int_{0}^{\cdot}(h \cdot M)_{u} d(1/L_{u}^{n}) + \int_{0}^{\cdot} 1/L_{u}^{n}h_{u} dM_{u} + \int_{0}^{\cdot} 1/L_{u}^{n}h_{u}^{\top} d\left<M,M\right>_{u}\lambda_{u}^{n} \nonumber\\
&= \int_{0}^{\cdot} 1/L_{u}^{n}\left((h \cdot M)_{u}\lambda_{u}^{n}
+ c_{u}c_{u}^{+}h_{u} \right) dS_{u},
\end{align}
where we have used \eqref{eq:orthogonalIm} in the last step.
Then, 
$h^{n} := 1/L^{n}((h \cdot M)\lambda^{n}+ cc^{+}h)1_{[0,\zeta]}$
belongs to $L(S)$.
\end{proof}

The following proposition states the price of the contingent claim in the presence of increasing profits
is that of the contingent claim that pays $\xi$ at time $\zeta$
only if the increasing profits is not realized.
\begin{proposition}
\label{prop:increasingprofit}
Suppose that Assumptions \ref{ass:trivial} and \ref{ass:prp} hold,
that $\widehat{K}_{0}^{\zeta} < \infty$
and that $\xi \widehat{L}_{\zeta}$ is integrable with respect to $P$.
Then, the price of the contingent claim that pays $\xi$ at time $\zeta$ is
\begin{equation}
E\left[\xi 1_{\{\zeta \le \rho^{*} \}} \widehat{L}_{\zeta}\right].\label{eq:price}
\end{equation}
\end{proposition}
\begin{proof}
Let $N_{t}^{n}:=E[\xi L_{\zeta}^{n} | \mathcal{F}_{t}]$
and $V_{t}^{n}:=E[\xi L_{\zeta}^{n} | \mathcal{F}_{t}]/L_{t\wedge \zeta}^{n}$
for each $n$.
Then, according to Assumption \ref{ass:prp} and Lemma \ref{lem:tradable},
there exists a predictable process $h^{n} \in L(S)$ such that
\begin{equation}
V_{t\wedge \zeta}^{n} = N_{0}^{n} + \int_{0}^{t\wedge \zeta} h_{u}^{n} dS_{u}.
\end{equation}
This formula represents an admissible trading strategy that replicates the payoff $\xi$,
because of $V_{t}^{n} \ge 0$ and $V_{\zeta}^{n}=\xi$.
The wealth process $V^{n}$ of this strategy 
converges decreasingly 
as $n \nearrow \infty$ by the monotone convergence theorem:
\begin{equation}
V_{t}^{n} = E\left[\xi \frac{\widehat{L}_{\zeta}}{\widehat{L}_{t \wedge \zeta}} \mathrm{e}^{-n (|\kappa|_{\zeta}-|\kappa|_{t})_{+}} \; \middle | \; \mathcal{F}_{t}\right]
\searrow
E\left[\xi \frac{\widehat{L}_{\zeta}}{\widehat{L}_{t \wedge \zeta}} 1_{\{\zeta \le \rho_{t}^{*} \}} \; \middle | \; \mathcal{F}_{t}\right],
\end{equation}
where $\rho_{t}^{*} := \inf \{u > t : |\kappa|_{u} > |\kappa|_{t} \}$.
The minimality follows from Proposition \ref{prop:minimality}.
\end{proof}

Jarrow and Protter \cite{Jarrow2005} proved the completeness of the market model in the presence of increasing profit
under the assumption 
that a unique probability $P^{*}$ equivalent to $P$ exists
such that 
the process $N:=S_{0}+M+ \lambda \cdot \left<M,M\right>$
is a $P^{*}$ local martingale (Jarrow and Protter \cite[Theorem 3.2]{Jarrow2005}).
In the proof, they demonstrated a replicating strategy with an initial fortune $E^{*}[\xi]$, where $E^{*}$ is the expectation operator with respect to $P^{*}$.
However, they did not discuss any optimality of the strategy
and their strategy does not take advantage of the existence of increasing profit.
By contrast, our strategies make use of increasing profit,
which may lead to smaller initial fortunes than those of Jarrow and Protter \cite{Jarrow2005}.
See Example \ref{ex:rgbm} below for a specific case.

We provide three examples of the financial market models comprising a
savings account and a single stock below.
Although Proposition \ref{prop:increasingprofit} does not state that there exists an admissible strategy that replicates the payoff with initial fortune \eqref{eq:price},
the optimal strategies are explicitly obtained in these specific cases.

In the first example, the stock price process is given as a reflected geometric Brownian motion,
which attracts much attention from the practical point of view (Buckner et~al. \cite{Buckner2022}).
The prices for call options are those of the corresponding barrier options that are knock-out at the reflection boundary.
We derive the price and replicating strategy using a partial differential equation argument without relying Proposition \ref{prop:increasingprofit}.
\begin{example}
\label{ex:rgbm}
Let $S$ be a reflected geometric Brownian motion with a lower reflecting boundary $b \in (0,S_{0})$.
More precisely, 
for a given Brownian motion $\beta$ with $\beta_{0}=0$,
and two constants $\mu$ and $\sigma \neq 0$,
the stock price process $S$ is a continuous adapted process that satisfies
$S_{t} \ge b$
and
\begin{equation}
S_{t} = S_{0} + \int_{0}^{t} \sigma S_{u} d\beta_{u} + \int_{0}^{t} \mu S_{u} du + L_{t},
\end{equation}
where 
$L$ is a continuous, adapted and nondecreasing process with $L_{0}=0$
and $\int_{0}^{t} 1_{\{S_{u} > b\}} dL_{u} = 0$.
Then, $L$ admits an expression $L_{t} = l_{t}^{b}/2$,
where $l_{t}^{b}$ is the local time of $S$ at $b$:
\begin{equation}
l_{t}^{b} := \lim_{\varepsilon \downarrow 0} \frac{1}{\varepsilon} \int_{0}^{t} 1_{(b,b+\varepsilon]}(S_{u}) d\left<S,S\right>_{u}.
\end{equation}
Buckner et~al. \cite{Buckner2022} pointed out that the financial market model $(1,S)$ violates the NIP condition (Buckner et~al. \cite[Proposition 3.1]{Buckner2022})
and then that the contingent claim valuations using the NA1 and NFLVR conditions cannot be applied.

We consider the price of a call option with strike $K$ and maturity $T$ for the stock.
Let $\tilde{S}$ be the geometric Brownian motion with no drift:
\begin{equation}
\tilde{S}_{t} = S_{0} + \int_{0}^{t} \sigma \tilde{S}_{u} d\beta_{u}
= S_{0} \mathrm{e}^{-\sigma^{2} t/2 + \sigma \beta_{t}}
\end{equation}
and
$C(\tau,s,b)$ be the price of the knock-out call option for the stock price process $\tilde{S}$ with $\tilde{S}_{0}=s$, knock-out level $b$, strike $K$ and maturity $\tau$:
\begin{equation}
C(\tau,s,b) := E\left[(\tilde{S}_{\tau}-K)_{+} : \tilde{H}_{b} > \tau \right],
\end{equation}
where $\tilde{H}_{b} = \inf \{t : \tilde{S}_{t} = b\}$.
Then, $C(\cdot,\cdot,b)$ is a solution to the partial differential equation:
\begin{align}
\left\{ \begin{array}{l}
C_{\tau} = \frac{1}{2} \sigma^{2} S^{2} C_{ss}, \\
C(\tau,s) = (s-K)_{+},\\
C(\tau,b) = 0. \label{eq:bse}
\end{array} \right.
\end{align}
An application of Ito formula to $C(T-t,S_{t},b)$ yields
\begin{equation}
(S_{T}-K)_{+}
= C(T-t,S_{t},b) + \int_{t}^{T} C_{s}(T-u,S_{u},b) dS_{u},
\end{equation}
which implies an admissible replicating strategy.
More precisely,
the self-financing strategy with initial fortune $C(T,S_{0},b)$
and $C_{S}(T-t,S_{t},b)$-unit of the stock at time $t \in [0,T]$
replicates the call option,
whose wealth process satisfies $C(T-t,S_{t},b) \ge 0$.
The minimality of the initial fortune follows from Proposition \ref{prop:minimality}.

Let $C(\cdot,\cdot,b,c)$ the solution to \eqref{eq:bse}
with the third condition replaced by $C(\tau,b,b) = -c$ for $c > 0$.
Then, the strategy based on $C(\cdot,\cdot,b,c)$
replicates the call option with smaller initial fortune than our price $C(\cdot,\cdot,b)$.
However, we rule out this strategy in Definition \ref{def:price},
otherwise the price could be arbitrarily small.

We remark that $C(\cdot,\cdot,0)$ is the Black--Scholes price for
the European call option without a barrier; that is, the solution to \eqref{eq:bse}
without the boundary condition $C(\tau,b,b) = 0$.
The strategy based on $C(\cdot,\cdot,0)$ does replicate the payoff,
which is the strategy proposed by Jarrow and Protter \cite{Jarrow2005}
(see the discussion following Jarrow and Protter \cite[Theorem 4.3]{Jarrow2005}).
The price proposed by Veestraeten \cite{Veestraeten2008}, which was based on a
wrong argument,
is a solution to \eqref{eq:bse} for which the boundary condition at $s=b$
is replaced with $C_{s}(\tau,b)=0$ (see (9) of Veestraeten \cite{Veestraeten2008}).
The strategy based on this price also replicates the payoff.
However, these prices are larger than ours $C(\cdot,\cdot,b)$.

The same argument as above can be applied to the case where $S$ satisfies
\begin{equation}
S_{t} = S_{0} + \int_{0}^{t} \sigma S_{u} d\beta_{u} - t + L_{t}
\end{equation}
with the reflection boundary $b=0$.
Then, the price of the call option at time $0$ is 
\begin{equation}
E[(\tilde{S}_{T}-K)_{+} 1_{\{\tilde{H}_{0} > T\}}]
=E[(\tilde{S}_{T}-K)_{+}] =  C(T,S,0),
\end{equation}
because of $P[\tilde{H}_{0} = \infty]=1$.
In this case, the presence of increasing profit
does not reduce the initial fortune for the replication
from the Black--Scholes price.
\end{example}

The second example is involved with a Brownian motion and its local time at the level $0$.
The level at which the stock price singularly behaves
decreases or increases at each time when the price hits that level.
This example contains the three-dimensional Bessel process as a special case.
In this example, the prices of European options are reduced to the knock-out option price under Bachelier's model.
The main difference from Example \ref{ex:rgbm} is 
that the position in the stock of the optimal replicating strategy has to be maintained
when the stock price spends in the boundary level.

\begin{example}
\label{ex:emery-perkins}
Let $S$ be defined as follows:
\begin{equation}
S_{t} = S_{0} + |\beta_{t}|-|\beta_{0}| + (\alpha-1) L_{t}^{0}(\beta)
= S_{0} + \int_{0}^{t} \mathrm{sign}(\beta_{u})d\beta_{u}+\alpha L_{t}^{0}(\beta), \label{eq:emery-perkins}
\end{equation}
where $\alpha \ge 1$ for simplicity,
$\beta$ is a Brownian motion with $|\beta_{0}| \le S_{0}$
and $L^{0}(\beta)$ is its local time at the level $0$.
We remark that $S$ is a reflected Brownian motion at the level $S=S_{0}-|\beta_{0}|$ if $\alpha=1$,
and $S$ is a three-dimensional Bessel process conditionally on $\inf_{t \ge 0} S_{t} =S_{0}-|\beta_{0}|$ if $\alpha=2$.
In addition,
the filtration generated by $\beta$ is strictly larger than that by $S$
if and only if $\alpha=2$ (see Revuz and Yor \cite[Exercise VI.3.18]{revuzyor}).

We consider the contingent claim that pays $g(S_{T})$ at time $T$.
According to Proposition \ref{prop:increasingprofit}, 
the price at time $t$ is 
\begin{equation}
V_{t} = E\left[g(S_{T})1_{\{T \le \rho_{t}^{*}\}}\; \middle | \; \mathcal{F}_{t} \right]
= W_{S_{t}}\left[g(X_{T-t})1_{\{T-t \le \tau_{j_{t}} \}} \right]
= \Psi^{g}(T-t,S_{t},j_{t}),
\end{equation}
where $j_{t} = S_{0}-|\beta_{0}|+(\alpha-1)L_{t}^{0}(\beta)$.
See Section \ref{sec:bachelier} for the notation.
Assume that $\Psi^{g}$ is sufficiently smooth in all variables.
Then, an application of Ito formula and \eqref{eq:heat} yield
\begin{align}
g(S_{T}) &= \Psi^{g}(T,S_{0},j_{0})
+ \int_{0}^{T} \Psi_{x}^{g}(T-u,S_{u},j_{u}) dS_{u}
+ \int_{0}^{T} \Psi_{j}^{g}(T-u,j_{u},j_{u}) dj_{u}\\
&= \Psi^{g}(T,S_{0},j_{0})
+ \int_{0}^{T} \left(\Psi_{x}^{g}(T-u,S_{u},j_{u})
+\frac{\alpha-1}{\alpha} \Psi_{j}^{g}(T-u,j_{u},j_{u}) 1_{\{\beta_{u}=0\}}\right)dS_{u},
\end{align}
which implies a replicating strategy.
\end{example}

The third example is a skew Brownian motion, which has been proposed for option pricing owing to its ability to reproduce empirical skewness of asset returns.
However, this model admits increasing profit as Rossello \cite{Rossello2012} pointed out.
In this example, the prices of European options are reduced to the knock-out price under Bachelier's model as the previous example.
The difference is that
the stock price can be above or below
the level where the stock price singularly behaves.
\begin{example}
Let $S$ be a skew Brownian motion with $\alpha \in (0,1) \setminus \{1/2\}$ for $b>0$; that is, $S$ is the unique adapted stochastic process that satisfies
\begin{equation}
S_{t} = S_{0} + \beta_{t} + (2\alpha -1) \tilde{L}_{t}(S),
\end{equation}
where $\tilde{L}(S)$ is the symmetric local time for $S$ at $0$
and $\beta$ is a Brownian motion with $\beta_{0}=0$.
Then, the time spent by $S$ on the level $0$ has Lebesgue measure zero,
because $|S|$ is a reflected Brownian motion.
Applications of Ito--Tanaka formula (Revuz and Yor \cite[Exercise VI.1.25]{revuzyor}) to $S_{t}^{+}:=\max(S_{t},0)$
and $S_{t}^{-}:=\max(-S_{t},0)$ yield
\begin{align}
S_{t}^{+} &= S_{0}^{+} + \int_{0}^{t}\left(1_{\{S_{u} > 0\}}+\frac{1}{2}1_{\{S_{u}=0\}}\right)dS_{u}+\frac{1}{2}\tilde{L}_{t}(S) \\
&= S_{0}^{+} + \int_{0}^{t}\left(1_{\{S_{u} > 0\}}+\left(\frac{1}{2}+\frac{1}{2(2\alpha-1)}\right)1_{\{S_{u}=0\}}\right)dS_{u}, \\
S_{t}^{-} &= S_{0}^{-} - \int_{0}^{t}\left(1_{\{S_{u} < 0\}}+\left(\frac{1}{2}-\frac{1}{2(2\alpha-1)}\right)1_{\{S_{u}=0\}}\right)dS_{u}.
\end{align}

We consider the contingent claim that pays $g(S_{T})$ at time $T$.
According to Proposition \ref{prop:increasingprofit}, 
the price at time $t$ is 
\begin{equation}
V_{t} = E\left[g(S_{T})1_{\{T < \rho_{t}^{*}\}}\; \middle | \; \mathcal{F}_{t} \right]
= W_{S_{t}}\left[g(X_{T-t})1_{\{T-t < \tau_{0} \}} \right]
= \Psi^{g}(T-t,S_{t},0).
\end{equation}
Let 
$\Psi^{\pm}(\tau,x) := \int_{\mathbb{R}_{\pm}} g(y)
(\mathrm{e}^{-\frac{(y-x)^{2}}{2\tau}}-\mathrm{e}^{-\frac{(y+x)^{2}}{2\tau}})dy$ for $\tau > 0$ and $x \in \mathbb{R}$.
Then an application of Ito--Tanaka formula to
$\Psi^{g}(T-t,S_{t})=\Psi^{+}(T-t,S_{t}^{+})+\Psi^{-}(T-t,-S_{t}^{-})$
and \eqref{eq:heat} yield
\begin{align}
g(S_{T})
&=
\Psi^{g}(T,S_{0},0)
+ \int_{0}^{T} \Psi_{x}^{+}(T-u,S_{u}^{+}) dS_{u}^{+}
- \int_{0}^{T} \Psi_{x}^{-}(T-u,-S_{u}^{-}) dS_{u}^{-}
\nonumber\\
&=
\Psi^{g}(T,S_{0},0)
+ \int_{0}^{T} \left(\Psi_{x}^{g}(T-u,S_{u},0)
+\frac{1}{2}\frac{1}{2\alpha-1} \Psi_{xx}^{g}(T-u,0,0) 1_{\{S_{u}=0\}}\right)dS_{u},
\end{align}
which implies a replicating strategy.
\end{example}

\subsection{Contingent claim valuation under strong arbitrage}
\label{sec:sa}
In this subsection, we consider the price of the contingent claim that pays $\xi$ at a finite stopping time $\zeta$ under strong arbitrage
with the conditions $\rho^{*} \ge \zeta$
and $\widehat{K}_{0}^{\zeta \wedge \rho^{+}} < \infty, \widehat{K}_{t+\rho^{+}}^{\zeta} < \infty$
for each $t > 0$,
which means that
a strong arbitrage opportunity occurs at most only at $\rho^{+}$ in the market model restricted on $[0,\zeta]$
and the NA1 condition holds on $[0,\zeta \wedge \rho^{+}]$ and $(\rho^{+},\zeta]$.
Then, the contingent claim valuation using the NA1 condition can be applied in the market model restricted on $[0,\zeta \wedge \rho^{+}]$.
We can focus on the situation with $\rho^{+}=0$ or $\zeta \le \rho^{+}$; that is, $\widehat{K}_{\varepsilon}^{\zeta}<\infty$.

Let $\widehat{Z}_{s,\cdot} = (\widehat{Z}_{s,t})_{t \ge 0}$ be
\begin{equation}
\widehat{Z}_{s,t}
:= \exp\left(-\int_{s \wedge t}^{t}\lambda_{u} dM_{u}-\frac{1}{2}\int_{s \wedge t}^{t} \lambda_{u}^{\top} d\left<M,M\right>_{u} \lambda_{u} \right) 1_{\{t < \rho_{s}^{1}\}},
\end{equation}
where $\rho_{s}^{1}:=\inf \{t > s : \widehat{K}_{s}^{t} = \infty\}$ for each $s > 0$,
and $\widehat{L}_{s,t} := \widehat{Z}_{s,t-}$.

\begin{proposition}
Suppose that Assumptions \ref{ass:trivial} and \ref{ass:prp} hold,
$\rho^{*} \ge \zeta$ and $\widehat{K}_{\varepsilon}^{\zeta}<\infty$,
and that $\xi \widehat{L}_{\varepsilon,\zeta}$ is integrable with respect to $P$ for each $\varepsilon > 0$.
Then, the price of the contingent claim that pays $\xi$ at time $\zeta$ is
\begin{equation}
\xi_{\rho^{+}} := \lim_{\varepsilon \downarrow 0} E \left[\xi \widehat{L}_{\varepsilon,\zeta} \; \middle | \; \mathcal{F}_{\varepsilon} \right],\label{eq:sa}
\end{equation}
provided the limit exists $P$-a.s.
\end{proposition}
\begin{proof}
Let $N_{t}^{\varepsilon} := E [\xi \widehat{L}_{\varepsilon,\zeta} | \mathcal{F}_{t} ]$ for $t \ge 0$ and each $\varepsilon \in (0,1] \cap \mathbb{Q}$
and
$V_{t} := E [\xi \widehat{L}_{t,\zeta} | \mathcal{F}_{t} ]$
for $t > 0$.
Then, according to Assumption \ref{ass:prp},
an $(\mathcal{F}_{t})$-predictable process $f^{\varepsilon} \in L(M)$
exists such that $N_{t}^{\varepsilon}-N_{0}^{\varepsilon}=(f^{\varepsilon} \cdot M)_{t}$.
A straightforward computation shows 
that an $(\mathcal{F}_{t})$-predictable process $h^{\varepsilon} \in L(S)$ exists
such that
\begin{equation}
\frac{N_{t\wedge \zeta}^{\varepsilon}}{\widehat{L}_{\varepsilon,t\wedge \zeta}}
= N_{\varepsilon \wedge \zeta}^{\varepsilon}
  +\int_{\varepsilon}^{t\wedge \zeta} h_{u}^{\varepsilon} dS_{u}
\end{equation}
for $t \ge \varepsilon$.
We remark that $\widehat{L}_{\varepsilon,\zeta}$ is positive 
by the assumption $\widehat{K}_{\varepsilon}^{\zeta}<\infty$.
This expression shows that $V$ has a continuous version in $(0,\zeta]$
and can be continuously extended in $[0,\zeta]$ by $V_{0}:=\xi_{\tau^{+}}$.
Hereafter, we assume that the stopped process $V^{\zeta}$ is continuous.

Let $\Theta_{t} :=V_{t}^{\zeta} - V_{0}$
and $h_{t}^{+} := \sum_{n=1}^{\infty} h_{t}^{1/n}1_{(1/(n+1),1/n]}(t)+h_{t}^{1}1_{(1,\infty)}(t)$.
Then, $\Theta$ is an $(\mathcal{F}_{t})$-adapted, continuous process
vanishing at time $0$  with 
\begin{equation}
\Theta_{t} = \lim_{n \rightarrow \infty} (V_{t}^{\zeta}-N_{1/n}^{1/n})
= \int_{0+}^{t \wedge \zeta} h_{u}^{+} dS_{u}.
\end{equation}
For each $\delta > 0$,
the pair $(V_{0}+\delta,h^{+}1_{[\tau_{\delta},\zeta]})$
with $\tau_{\delta}:=\inf\{t : |\Theta_{t}| > \delta \}$
is an admissible superreplicating strategy,
which follows from
\begin{equation}
(V_{0}+\delta) + \int_{t \wedge \tau_{\delta}}^{t\wedge \zeta} h_{u}^{+} dS_{u}
= V_{0}+\Theta_{t} + (\delta - \Theta_{t \wedge \tau_{\delta}})
\ge V_{0}+\Theta_{t} 
= V_{t}^{\zeta}
\ge 0
\end{equation}
and $V_{\zeta}(\omega)= N_{\zeta}^{\varepsilon}(\omega) / \widehat{L}_{\varepsilon,\zeta}(\omega) =\xi(\omega)$
for $P$-a.s. $\omega \in \Omega$,
where $\varepsilon \in (0,\zeta(\omega)] \cap \mathbb{Q}$ is arbitrarily taken for $\omega$ with $\zeta(\omega)>0$
and $\varepsilon=0$ for $\omega$ with $\zeta(\omega)=0$.

Suppose that $a \ge 0$ and $h \in L(S)$
satisfies $\xi \le V_{\zeta}(a,h)$ and $V_{t}(a,h)\ge 0$ for $t \ge 0$.
For $\varepsilon>0$,
let $\tilde{N}_{t}^{\varepsilon}:= (a + (h \cdot S)_{t})\widehat{L}_{\varepsilon,t}$ for $t > 0$.
Then, $\tilde{N}^{\varepsilon}$ is a nonnegative, continuous supermartingale in $[\varepsilon,\zeta]$.
Hence, we obtain $\xi_{\rho^{+}} \le a$ by letting $\varepsilon \downarrow 0$ in
\begin{equation}
N_{\varepsilon}^{\varepsilon}
= E \left[\xi \widehat{L}_{\varepsilon,\zeta} \; \middle | \; \mathcal{F}_{\varepsilon} \right]
\le E \left[\tilde{N}_{\zeta}^{\varepsilon} \; \middle | \; \mathcal{F}_{\varepsilon} \right]
\le \tilde{N}_{\varepsilon}^{\varepsilon}
= a + (h \cdot S)_{\varepsilon}.
\end{equation}
\end{proof}

We investigate the case of the three-dimensional Bessel process.
\begin{example}
\label{ex:bes3}
Let the stock price process $S$ be a three-dimensional Bessel process with $S_{0} \ge 0$,
for which we consider the contingent claim that pays $g(S_{T})$ at time $T$.
This model with $S_{0} > 0$ violates the NA condition
and satisfies the NA1 condition with the num\'{e}raire portfolio $S/S_{0}$.
Then, the contingent claim valuation using the NA1 condition, or the Benchmark approach, can be applied.
The price of the contingent claim is 
\begin{equation}
E\left[g(S_{T})\frac{S_{0}}{S_{T}} \; \middle | \; \mathcal{F}_{t} \right]
= P_{S_{t}}^{(1/2)}\left[g(X_{T-t})\frac{X_{0}}{X_{T-t}}\right]
= W_{S_{t}}[g(X_{T-t}): \tau_{0} > T-t]
= \Psi^{g}(T-t,S_{t},0).
\end{equation}

Suppose that $S_{0} = 0$.
Then, a strong arbitrage occurs at time $0$; that is, $\rho^{+}=0$.
The price at time $0$ is
\begin{equation}
\lim_{t \downarrow 0} E \left[g(S_{T}) \widehat{L}_{t,T} \; \middle | \; \mathcal{F}_{t} \right]
=\lim_{t \downarrow 0} \Psi^{g}(T-t,S_{t},0)
= \Psi^{g}(T,S_{0},0) = 0.
\end{equation}

Next, we consider the case of $S_{0} > 0$, conditionally on $J_{0}=j \in (0,S_{0})$.
A strong arbitrage occurs at $\rho^{+} = \inf \{t > 0: S_{t} = j\}$,
and $S-j$ is a three-dimensional Bessel process after $\rho^{+}$ (see Section \ref{sec:williams}).
Then, the price $\xi_{\rho^{+}}$ at time $\rho^{+} < T$ is reduced to the case with $S_{0}$; that is, $\xi_{\rho^{+}}=0$.
The price at time $t < T \wedge \rho^{+}$ of the contingent claim under the consideration
is that of the contingent claim that pays $g(S_{T})1_{\{T \le \rho^{+}\}}$:
\begin{equation}
E \left[g(S_{T})1_{\{T \le \rho^{+}\}}\frac{S_{0}}{S_{T}} \; \middle | \; \mathcal{F}_{t} \right] 
= P_{S_{t}}^{(1/2)}\left[g(X_{T-t})\frac{X_{0}}{X_{T-t}} : \tau_{j} > T-t \right]
= \Psi^{g}(T-t,S_{t},j).
\end{equation}
The trading strategy in this example is also valid 
in the financial market model in Example \ref{ex:emery-perkins} with $\alpha = 2$ and $J_{0}=j$,
where the price at time $0$ is also $\Psi^{g}(T,S_{0},j)$.
This shows that the optimal strategies are not unique, while the optimal price is.
\end{example}

In the following example, we review Jarrow and Protter \cite[Section 7]{Jarrow2005}, which treated this problem
using Delbaen and Schachermayer \cite[Example 3.4]{DelbaenSchachermayer1995b}
and pointed out 
that a market model with strong arbitrage can be complete
in the sense that payoffs are restricted to a certain class.
\begin{example}
\label{ex:jarrowprotter}
Let the stock price $S$ be 
\begin{equation}
S_{t} = S_{0} + \beta_{t} + \sqrt{t}.
\end{equation}
This market model admits a strong arbitrage at time $\rho^{+}=0$.
However, if the model is restricted on $[\varepsilon,T]$ for each $\varepsilon>0$,
a unique equivalent local martingale measure $Q^{\varepsilon} := \widehat{L}_{\varepsilon,T} P$
exists, under which 
$S$ is a Brownian motion on $[\varepsilon,T]$,
where $\widehat{L}_{\varepsilon,T}$ is defined as
\begin{equation}
\widehat{L}_{\varepsilon,T}:= \exp
  \left(-\frac{1}{2}\int_{\varepsilon}^{T}\frac{d\beta_{u}}{\sqrt{u}}
    - \frac{1}{8}\int_{\varepsilon}^{T}\frac{du}{u} \right).
\end{equation}

Consider the contingent claim pays $\xi \in \mathcal{F}_{T}$ at time $T$.
Jarrow and Protter \cite{Jarrow2005} assumed that $E[\xi \widehat{L}_{\varepsilon,T}]$ converges as $\varepsilon \downarrow 0$,
which slightly differs from ours.
For each $\varepsilon > 0$,
using the martingale representation for $S$ under $Q^{\varepsilon}$, 
the pair of 
$\Lambda_{\varepsilon} \in \mathcal{F}_{\varepsilon}$
and $h^{\varepsilon} \in L(S)$
exists such that 
\begin{equation}
\xi = \Lambda_{\varepsilon} + \int_{\varepsilon}^{T} h_{u}^{\varepsilon} dS_{u}.
\end{equation}
The uniqueness shows
that the sequence $\{(\Lambda_{\varepsilon},h^{\varepsilon})\}$
determines a replicating strategy with an initial fortune
$\lim_{\varepsilon \downarrow 0}E[\xi \widehat{L}_{\varepsilon,T}]$.
The price and replication are reduced to those under Bachelier's model.

We caution about the integrability.
Let $f \in L^{2}(I,dt)$ with $I:=[0,1/2]$ such that 
$\int_{0}^{1/2} f(t)dt/\sqrt{t} = \infty$.
Then, $h_{t}(\omega) := \beta_{t}(\omega)f(t)/\sqrt{t}$
belong to $L^{2}(I,dt) \cap L^{1}(I,dt)^{c}$,
but $\int_{0+}^{1/2}h_{t}(\omega)dt/\sqrt{t}:=\lim_{\varepsilon \downarrow 0}\int_{\varepsilon}^{1/2}h_{t}(\omega)dt/\sqrt{t}$
is well defined for $P$-a.s. $\omega \in \Omega$
(see Chaumont and Yor \cite[6.1.3]{chaumont2012}, where an example of $f$ is
given as $f(t)=t^{-1/2} (\log 1/t)^{-\alpha}$ with $\alpha \in (1/2,1)$).
Then, 
\begin{equation}
\eta_{t} := 1 + \int_{0}^{t} h_{u} d\beta_{u} + \int_{0+}^{t} h_{u} \frac{du}{\sqrt{u}}
\end{equation}
is an adapted continuous process.
Consider the contingent claim that pays $\xi:=\eta_{\sigma}$
at time $T=1/2$, where $\sigma:=\inf \{t \in [0,1/2]: \eta_{t} \notin (0,2)\}$.
For each $\varepsilon > 0$, we obtain
\begin{equation}
E[\xi \widehat{L}_{\varepsilon,T}]-E[\eta_{\sigma \wedge \varepsilon}\widehat{L}_{\varepsilon,T}] 
= 
E[(\eta_{\sigma}-\eta_{\sigma \wedge \varepsilon})\widehat{L}_{\varepsilon,T}] 
= E^{\varepsilon}[(\eta_{\sigma}-\eta_{\sigma \wedge \varepsilon})] 
= 0,
\end{equation}
where $E^{\varepsilon}$ denotes the expectation operator with respect to $Q^{\varepsilon}$.
Hence, $\lim_{\varepsilon \downarrow 0}E[\xi \widehat{L}_{\varepsilon,T}]=1$
follows from that
\begin{equation}
E[\xi \widehat{L}_{\varepsilon,T}]
= E[\eta_{\sigma \wedge \varepsilon}\widehat{L}_{\varepsilon,T}]
=E\left[\eta_{\sigma \wedge \varepsilon}E\left[\widehat{L}_{\varepsilon,T} \; \middle | \; \mathcal{F}_{\varepsilon}\right] \right]
=E[\eta_{\sigma \wedge \varepsilon}]
\end{equation}
converges to $1$ as $\varepsilon \downarrow 0$, according to the bounded convergence theorem.
However, the pair $(1,h1_{[0,\sigma]})$ is not an admissible strategy, because of $h1_{[0,\sigma]} \notin L(S)$.
\end{example}

\subsection{Contingent claim valuation under arbitrage of the first kind}
\label{sec:a1}
In this subsection, we consider the price of the contingent claim
that pays $\xi$ at a finite stopping time $\zeta$ under arbitrage of the first kind with the condition $\rho^{*},\rho^{+} \ge \zeta$.
\begin{proposition}
Suppose that Assumptions \ref{ass:trivial} and \ref{ass:prp} hold, and $\rho^{*},\rho^{+} \ge \zeta$ holds.
Then, the price of the contingent claim that pays $\xi$, for which we assume $0 \le \xi < 1$, at time $\zeta$
is 
\begin{equation}
E \left[\xi \widehat{L}_{\zeta} 1_{\{\zeta \le \rho^{1} \}}\right].
\end{equation}
\end{proposition}
\begin{proof}
Recalling Lemma \ref{lem:L},
we obtain $\varepsilon /\widehat{L}\lambda 1_{[0,\theta_{\varepsilon}^{1}]} \in \mathcal{A}_{\varepsilon}$
and 
\begin{equation}
\varepsilon + \int_{0}^{\zeta} \frac{\varepsilon}{\widehat{L}_{u}} \lambda_{u} 1_{\{u \le \theta_{\varepsilon}^{1}\}}dS_{u}
= \frac{\varepsilon}{\widehat{L}_{\theta_{\varepsilon}^{1} \wedge \zeta}}
\ge 1_{\{\theta_{\varepsilon}^{1} < \zeta\}}
\ge \xi 1_{\{\theta_{\varepsilon}^{1} < \zeta\}}.
\end{equation}
Let $N_{t}^{1} := E [\xi \widehat{L}_{\zeta} 1_{\{\zeta \le \rho^{1} \}} |\mathcal{F}_{t}]$ for $t \ge 0$
be a nonnegative martingale.
Then, according to Assumption \ref{ass:prp},
a predictable integrand $f \in L(M)$ exists such that
$N_{t}^{1}=N_{0}^{1}+(f\cdot M)_{t}$ for $t \ge 0$.
As in the proof of Lemma \ref{lem:tradable}, a predictable integrand $h \in L(S)$
exists such that $(N^{1}/\widehat{L})^{\theta_{\varepsilon}^{1}\wedge \zeta} = N_{0}^{1} + (h \cdot S)^{\theta_{\varepsilon}^{1}\wedge \zeta}$.
Then, we obtain
\begin{equation}
N_{0}^{1} + \int_{0}^{\zeta} h_{u} 1_{\{u \le \theta_{\varepsilon}^{1} \}}dS_{u}
\ge \xi 1_{\{\zeta \le \theta_{\varepsilon}^{1}\}}
\end{equation}
because of $\theta_{\varepsilon}^{1} \le \rho^{1}$.
The combined admissible strategy 
$(N_{0}^{1}+\varepsilon,(h+\varepsilon /\widehat{L}\lambda)1_{[0,\theta_{\varepsilon}^{1}\wedge \zeta]})$
superreplicates the payoff $\xi$.
The minimality follows from Proposition \ref{prop:minimality}.
\end{proof}

If the nonnegative local martingale $\widehat{L}$ is a uniformly integrable martingale,
a probability measure $Q$ can be defined as $Q = \widehat{L}_{\zeta} P$ on $\mathcal{F}_{\zeta}$,
under which the stock price process $S$ is a local martingale.
Then, the price is expressed with
$V_{t}= E^{Q} [\xi 1_{\{\zeta \le \rho_{t}^{1} \}} \; | \; \mathcal{F}_{t} ]$,
where $E^{Q}$ is an expectation operator with respect to $Q$.
In case of $Q[\zeta \le \rho^{1}]=1$,
the arbitrage opportunity of the first kind does not contribute to
reduction of the initial fortune to replicate the contingent claim.

We provide the following two examples, for which the NA1 condition is violated
(see Criens and Urusov \cite[Theorem 3.10]{criens2024} for the criteria,
which is shortly reviewed in Section \ref{sec:bankruptcy}).
\begin{example}
Let $S$ be a Bessel process stopped at $S_{t}=0$ with index $\nu < 0$; that is,
$S$ satisfies
\begin{equation}
S_{t} = S_{0} + \beta_{t} +\left(\nu+\frac{1}{2}\right)\int_{0}^{t} \frac{du}{S_{u}},
\end{equation}
until $S$ hits $0$,
where $\beta$ is a Brownian motion with $\beta_{0}=0$.
Then, an arbitrage of the first kind occurs at $\rho^{1} = \inf \{t : S_{t}=0 \} < \infty$.
The price at time $t$ of the contingent claim that pays $g(S_{T})$ at time $T$ is
\begin{align}
V_{t} &= E\left[g(S_{T})\widehat{L}_{t,T}1_{\{T < \rho_{t}^{1} \}} \; \middle | \; \mathcal{F}_{t}\right]\nonumber\\
&= P_{S_{t}}^{(\nu)}\left[g(X_{T-t})\left(\frac{X_{0}}{X_{T-t}}\right)^{\nu+1/2} 
\exp\left(\frac{1}{2}\left(\nu^{2}-\frac{1}{4}\right)\int_{0}^{T-t}\frac{du}{X_{u}^{2}}\right)\right]\nonumber\\
&= W_{S_{t}}\left[g(X_{T-t}) : T-t <\rho_{0}\right],
\end{align}
conditionally on $\{t < \rho^{1} \}$,
where we have used the absolute continuity relationship \eqref{eq:bessel_ac}.
\end{example}

\begin{example}
\label{ex:constdrift}
Let the stock price process $S$ satisfy
\begin{equation}
S_{t} = S_{0} + \int_{0}^{t} S_{u} d\beta_{u} -t,
\end{equation}
where we assume that $0$ is an absorption boundary.
Then, an arbitrage of the first kind occurs at $\rho^{1} = \inf \{t : S_{t}=0 \} < \infty$.

Let $\widehat{L}_{t}= \exp\left(\int_{0}^{t} d\beta_{u}/S_{u}-\frac{1}{2}\int_{0}^{t}du/S_{u}^{2}\right)1_{\{t \le \rho^{1}\}}$
and $\tau_{n}=\inf \{t : S_{t} < 1/n\}$ for $n \in \mathbb{N}$.
Then, the stopped process $\widehat{L}^{\tau_{n}}$ is a $P$-martingale.
The probability measure $Q$ can be defined by 
$Q = \widehat{L}_{T}^{\tau_{n}}P$ on $\mathcal{F}_{T \wedge \tau_{n}}$,
under which $S^{\tau_{n}}$ is a stopped geometric Brownian motion.
Furthermore, $Q$ can be extended on $\mathcal{F}_{T}$
because of
\begin{equation}
E \left[\widehat{L}_{T}\right]
= \lim_{n \rightarrow \infty} E \left[\widehat{L}_{T} : T < \tau_{n}\right]
= \lim_{n \rightarrow \infty} E^{Q}\left[\inf_{t \in [0,T]} S_{t} > 1/n \right]
= 1.
\end{equation}
Then, the price at time $t$ of the contingent claim that pays $g(S_{T})$ at
time $T$ is that of the Black--Scholes model:
\begin{equation}
V_{t} = E\left[g(S_{T})\widehat{L}_{t,T} \; \middle | \; \mathcal{F}_{t}\right]
= W_{\log S_{t}} \left[g\left(\mathrm{e}^{X_{T-t}-(T-t)^{2}/2}\right)\right],
\end{equation}
conditionally on $\{t < \rho^{1} \}$.
\end{example}

\subsection{Proof of Theorem \ref{thm:main}}
\label{sec:proof}
We assume without loss of generality that $\xi < 1$.
Once we demonstrate a superreplicating strategy with the initial fortune 
less than $3\varepsilon + \xi_{0}$ for arbitrary $\varepsilon > 0$,
Theorem \ref{thm:main} follows from Proposition \ref{prop:minimality}.

Let $\theta := \rho^{*} \wedge \rho^{+} \wedge \theta_{\varepsilon}^{1}$.
We remark that $\theta \le \rho$ always holds
and 
\begin{align}
& \theta_{\varepsilon}^{1}(\omega) = \rho^{1}(\omega) < \infty \iff \theta_{\varepsilon}^{1}(\omega) = \rho^{1}(\omega) = \rho^{+}(\omega) < \infty,\\
& \theta(\omega) = \rho(\omega) < \infty
\implies \rho(\omega) = \rho^{*}(\omega) < \infty
\text{ or } \rho(\omega) = \rho^{+}(\omega) < \infty.
\end{align}
We decompose the sample space $\Omega$ into 
$\{\theta < \rho \wedge \zeta\}$
and 
\begin{align}
\{\rho \wedge \zeta \le \theta\}
&= \{\zeta \le \theta \wedge \rho \} \cup \{\theta = \rho < \zeta\} \\
&= \{\zeta \le \theta \} 
\cup \left(\{\rho = \rho^{+} < \zeta\} \cap \{\rho \le \theta_{\varepsilon}^{1}\}\right)
\cup \left(\{\rho^{+} \neq \rho = \rho^{*} < \zeta\} \cap \{\rho \le \theta_{\varepsilon}^{1}\}\right)
\end{align}
and $\xi$ into $\xi = \xi^{\prime} + \xi_{+} + \xi_{*} + \xi_{1}$,
where
\begin{equation}
\xi^{\prime} = \xi 1_{\{\zeta \le \theta \}}, \;
\xi_{+} = \xi 1_{\{\rho = \rho^{+} < \zeta\} \cap \{\rho \le \theta_{\varepsilon}^{1}\}}, \;
\xi_{*} = \xi 1_{ \{\rho^{+} \neq \rho = \rho^{*} < \zeta\} \cap \{\rho \le \theta_{\varepsilon}^{1}\}}, \;
\xi_{1} = \xi 1_{\{\theta < \rho \wedge \zeta\}}.
\end{equation}
In the strategies explained below,
the first two components are superreplicated with the initial fortune less than $\xi_{0} + \varepsilon$
and the rest is with $2\varepsilon$.
\begin{description}
\item[1)] 
To superreplicate the payoff $\xi_{1}$,
we use the strategy $(\varepsilon,\varepsilon /\widehat{L} \lambda 1_{[0,\theta]}) \in \mathbb{R}_{+} \times \mathcal{A}_{\varepsilon}$:
\begin{equation}
\varepsilon + \int_{0}^{\zeta} \frac{\varepsilon}{\widehat{L}_{u}} \lambda_{u} 1_{\{u\le \theta\}}dS_{u}
= \frac{\varepsilon}{\widehat{L}_{\theta  \wedge \zeta}}
= \frac{\varepsilon}{\widehat{L}_{\zeta}} 1_{\{\zeta \le \theta\}}
+ \frac{\varepsilon}{\widehat{L}_{\theta}} 1_{\{\theta = \rho < \zeta\}}
+ 1_{\{\theta < \rho \wedge \zeta\}}
\ge \xi_{1}.
\end{equation}
\item[2)]  
To superreplicate the payoff $\xi_{*}$,
we consider a superreplicating strategy for the contingent claim 
that pays $1_{A_{*}}$ with $A_{*}:=\{\rho^{+} \neq \rho = \rho^{*} < \zeta\} \cap \{\rho \le \theta_{\varepsilon}^{1}\}$
at time $\zeta^{\prime} := \zeta \wedge \rho^{+} \wedge \theta_{\varepsilon/2}^{1}$.
We remark that
$A_{*}=\{\rho^{+} \neq \theta = \rho^{*} < \zeta\} \cap \{\theta = \rho\} \in \mathcal{F}_{\zeta \wedge \theta} \subseteq \mathcal{F}_{\zeta^{\prime}}$.
Then, we obtain $\rho^{*} < \zeta^{\prime} = \zeta \wedge \rho^{+}$
on $\{\theta_{\varepsilon}^{1}=\theta_{\varepsilon/2}^{1}\} \cap A_{*}$
and $\rho^{*} < \zeta^{\prime} = \zeta \wedge \theta_{\varepsilon/2}^{1} $ on $\{\theta_{\varepsilon}^{1}\neq \theta_{\varepsilon/2}^{1}\} \cap A_{*}$;
that is, $P[A_{*} \cap \{\zeta^{\prime} \le \rho^{*}\}]=0$.
As in the proof of Proposition \ref{prop:increasingprofit},
$h^{n} \in L(S)$ exists such that
\begin{equation}
\xi_{0}^{n}+ \int_{0}^{\zeta^{\prime}} h_{u}^{n} dS_{u}
=1_{A_{*}},
\end{equation}
where $\xi_{0}^{n} :=E[1_{A_{*}} L_{\zeta^{\prime}}^{n} | \mathcal{F}_{0}]$
and $(L_{t}^{n})_{t \ge 0}$ is defined as in Section \ref{sec:ip}.
The conditional expectation $\xi_{0}^{n}$ decrease to $E[1_{A_{*}} L_{\zeta^{\prime}}1_{\{\zeta^{\prime} \le \rho^{*}\}} | \mathcal{F}_{0}] = 0$ as $n \rightarrow \infty$, $P$-a.s.
Let $\Omega_{n}:=\{\omega : \inf\{k : N_{0}^{k}(\omega) < \varepsilon\}=n\} \in \mathcal{F}_{0}$.
Then, 
the strategy $(\varepsilon, \sum_{n=1}^{\infty}h^{n}1_{\Omega_{n}}1_{[0,\zeta^{\prime}]}) \in \mathbb{R}_{+} \times \mathcal{A}_{\varepsilon}$
superreplicates the payoff $\xi_{*}$.
\item[3)]  
To superreplicate the payoff $\xi_{+}$,
we consider a superreplicating strategy for the contingent claim 
that pays $\xi_{\rho^{+}}^{\varepsilon}:=(\varepsilon+\xi_{\rho^{+}})1_{\{\rho = \rho^{+} < \zeta\} \cap \{\rho \le \theta_{\varepsilon}^{1}\}}$ at time $\rho^{+} \wedge \zeta$.
We remark that $\xi_{\rho^{+}}^{\varepsilon}$ is $\mathcal{F}_{\rho^{+} \wedge \zeta}$-measurable.
Let 
\begin{equation}
N_{t} := E[\xi_{\rho^{+}}^{\varepsilon}\widehat{L}_{\rho^{+}} | \mathcal{F}_{t}]
\end{equation}
and $V_{t}:=N_{t}/\widehat{L}_{t}^{\theta_{\varepsilon}^{1}}$ for $t \ge 0$.
Then, $h \in L(S)$ exists such that $V_{t}^{\theta_{\varepsilon}^{1}} = N_{0} + (h\cdot S)_{t}^{\theta_{\varepsilon}^{1}}$.
We obtain
\begin{equation}
N_{0} + \int_{0}^{\zeta \wedge \rho^{+}} h_{u}1_{[0,\theta_{\varepsilon}^{1}]}(u)dS_{u}
= V_{\zeta \wedge \rho^{+} \wedge \theta_{\varepsilon}^{1}}
\ge \xi_{\rho^{+}}^{\varepsilon}1_{\{\zeta \wedge \rho^{+} \le \theta_{\varepsilon}^{1}\}}
= \xi_{\rho^{+}}^{\varepsilon}.
\end{equation}
The payoff $\xi_{\rho^{+}}^{\varepsilon}$ at time $\rho^{+} \wedge \zeta$
is superreplicated with the initial fortune $N_{0}$.
On the event $\{\rho^{+} < \zeta\}$,
we use $\xi_{\rho^{+}}^{\varepsilon}$ at time $\rho^{+}$ as the initial fortune 
and an appropriate strategy to superreplicate the payoff $\xi_{+}$ at time $\zeta$.
Finally, the initial fortune $N_{0}$ is less than the second term of \eqref{eq:main} plus $\varepsilon$:
\begin{equation}
N_{0} 
\le \varepsilon + E[\xi_{\rho^{+}}\widehat{L}_{\rho^{+}}^{\theta_{\varepsilon}^{1}} 1_{\{\rho = \rho^{+} < \zeta\} \cap \{\rho \le \theta_{\varepsilon}^{1}\}}| \mathcal{F}_{0}]
\le \varepsilon + E[\xi_{\rho^{+}}\widehat{L}_{\rho^{+}} 1_{\{\rho = \rho^{+} < \zeta\}}| \mathcal{F}_{0}].
\end{equation}
\item[4)]  
To superreplicate the payoff $\xi^{\prime}$,
we consider a superreplicating strategy for the contingent claim 
that pays $\xi^{\prime}$ at time $\zeta$.
Let 
\begin{equation}
N_{t} := E\left[\xi^{\prime}\widehat{L}_{\zeta} \middle | \; \mathcal{F}_{t}\right]
\end{equation}
and $V:=N/\widehat{L}^{\theta_{\varepsilon}^{1}}$.
Then, $h \in L(S)$ exists such that
such that $V_{t}^{\theta_{\varepsilon}^{1}} = N_{0} + (h\cdot S)_{t}^{\theta_{\varepsilon}^{1}}$.
Then, we obtain
\begin{equation}
N_{0} + \int_{0}^{\zeta} h_{u}1_{[0, \theta_{\varepsilon}^{1}]}(u)dS_{u}
= V_{\zeta \wedge \theta_{\varepsilon}^{1}}
\ge \xi^{\prime} 1_{\{\zeta \le \theta_{\varepsilon}^{1}\}}
= \xi^{\prime}.
\end{equation}
Finally, the initial fortune $N_{0}$ is less than the first term of \eqref{eq:main}:
\begin{equation}
N_{0}= E\left[\xi^{\prime}\widehat{L}_{\zeta} \middle | \; \mathcal{F}_{0}\right]
= E\left[\xi \widehat{L}_{\zeta}1_{\{\zeta \le \theta\}} \middle | \; \mathcal{F}_{0}\right]
\le E\left[\xi \widehat{L}_{\zeta}1_{\{\zeta \le \rho\}} \middle | \; \mathcal{F}_{0}\right].
\end{equation}
\end{description}

\section{Modeling stock price with funding}
\label{sec:funding}
In this section, we model the stock price process \eqref{eq:underlying} with $d=1$:
\begin{equation}
S_{t} = S_{0} + M_{t} + \int_{0}^{t} \lambda_{u} d\left<M,M\right>_{u} + \kappa_{t}, \label{eq:underlying1d}
\end{equation}
incorporating effects of funding.
We regard the fundraiser and arbitrager as large traders
in the sense that they affect the stock price through their funding
and arbitrage strategies respectively,
and the rest of them as small traders in the sense that they do not.

In the literature, the singular term $\kappa$ is used to express the effect caused by artificial price support by a large trader.
For example, Jarrow and Protter \cite{Jarrow2005} began with a semimartingale without a singular term
and supposed that a level exists at which a large trader buys or sells.
The price impact is modeled in the discrete time
and then converges to the local time of the semimartingale at that level.
Finally, they added this singular term to the original semimartingale (see Jarrow and Protter \cite[Section 4]{Jarrow2005} for more detail).

We also regard a singular term as a result caused by large traders.
Our ingredient is to introduce the number of the stocks 
and a constraint so that the market participants play a zero-sum game.
If no arbitrager exists, the fundraiser can control the singular term through the number of the stocks, which is demonstrated in Section \ref{sec:in the absence of arbitragers}.
The existence of an arbitrager makes the situation so complicated
that we discuss the optimality of arbitrage strategies under simplified assumptions through examples in Section \ref{sec:in the presence of arbitragers}.

The presence of arbitragers is not desirable to the fundraiser from the perspective of the market capitalization in this market model.
We propose funding strategies that cause no increasing profit to the market participants except the fundraiser in Section \ref{sec:fs_and_process}.
More precisely, increasing profits exist for the fundraiser,
but it cannot be observed by agents whose information is generated by the stock price only.
Then, we investigate whether the funding strategies cause other kinds of arbitrage 
and 
a stock price to be $0$, which we interpret as bankruptcy,
in Section \ref{sec:bankruptcy}.

\subsection{Modeling stock price under a zero-sum constraint}
\label{sec:zero-sum}
Let $n=(n_{t})_{t\ge 0}$ be a nonnegative continuous predictable process of finite variation, 
which represents the number of stocks at time $t \ge 0$.
At time $t$ with $dn_{t} > 0$, the fundraiser (``he'') issues $dn_{t}$ units of new stocks and raises $S_{t}dn_{t}$.
On the contrary, at time $t$ with $dn_{t} < 0$, he repurchases $-dn_{t}$ units of stocks and spends $-S_{t}dn_{t}$.
As a result, the raised money by him up to $t$ is $F_{t} := \int_{0}^{t} S_{u} dn_{u}$.

Corporate stock issuance and repurchase plans affect its stock price.
We interpret $d\kappa_{t}$ as the impact caused by these funding activities.
More precisely, we assume that $d\kappa_{t}$ is supported on $\{t : dn_{t} \neq 0\}$.
The aggregate impact on the market by these activities is $H_{t} := \int_{0}^{t} n_{u} d\kappa_{u}$.

\subsubsection{In the absence of arbitragers}
\label{sec:in the absence of arbitragers}
Under the assumption that no arbitrager exists,
we introduce the following constraint on $F$ and $H$,
which means that
the fundraiser and rest of the market participants play a zero-sum game at each time $t$ with $dn_{t} \neq 0$:
\begin{assumption}
\label{ass:constraint1}
The processes $F$ and $H$ satisfy $F_{t} + H_{t} = 0$ for all $t \ge 0$.
\end{assumption}

Under this constraint, the singular term $\kappa$ is determined by
$d\kappa_{t} = -S_{t} \frac{dn_{t}}{n_{t}}$, or equivalently, by
\begin{equation}
n_{t} = n_{0} \exp \left(-\int_{0}^{t}\frac{d\kappa_{u}}{S_{u}}\right).
\end{equation}
In other words, the fundraiser can control the singular term $\kappa$
as he likes through issuing or repurchasing the stocks.
As a result, the market capitalization $V_{t} := n_{t}S_{t}$ is a semimartingale without singular terms:
\begin{equation}
V_{t} = V_{0} + \int_{0}^{t} n_{u}dS_{u} + \int_{0}^{t} S_{u} dn_{u}
= V_{0} + \int_{0}^{t} n_{u}d(S_{u}-\kappa_{u}).\label{eq:capital}
\end{equation}

\begin{example}
\label{ex:emery-perkins2}
Let the stock price $S$ be defined by \eqref{eq:emery-perkins} in Example \ref{ex:emery-perkins}; that is,
\begin{equation}
S_{t} = S_{0} + |\beta_{t}|-|\beta_{0}| + (\alpha-1) L_{t}^{0}(\beta)
= S_{0} + \int_{0}^{t} \mathrm{sign}(\beta_{u})d\beta_{u}+\alpha L_{t}^{0}(\beta),\label{eq:emery-perkins2}
\end{equation}
where $\alpha \ge 1$ for simplicity,
$\beta$ is a Brownian motion
and $L^{0}(\beta)$ is its local time at the level $0$.

In our interpretation, the fundraiser repurchases the stocks 
at the price $S_{t}=b + (\alpha-1) L_{t}^{0}(\beta)$ with $b = S_{0}-|\beta_{0}|$.
The number $n$ of the stocks decreases by $-dn_{t}=\alpha n_{t}dL_{t}/S_{t}$ at time $t$.
Then, we obtain
\begin{equation}
n_{t} = n_{0} \left(1+\frac{c}{b} L_{t} \right)^{-\left(1+\frac{1}{c}\right)},
\end{equation}
where $c=\alpha-1$ and $(1+ax)^{1/x} :=\mathrm{e}^{a}$ for $x=0$.
The money that he spends to repurchase the stocks up to the time $t$ is
\begin{equation}
-F_{t} = (1+c)n_{0}
\left(1-\left(1+\frac{c}{b}L_{t}\right)^{-\frac{1}{c}}\right) b.
\end{equation}
\end{example}

\begin{example}
\label{ex:doubly-reflected-BM}
Let the stock price $S$ be a doubly reflected Brownian motion
$S = \varphi(B)$ for a Brownian motion $B$,
where the function $\varphi : \mathbb{R} \longrightarrow [K_{1},K_{2}]$ for $K_{1}, K_{2} \ge 0$ with $\Delta := K_{2}- K_{1}>0$
satisfies 
$\varphi(2n \Delta)=K_{1}$, $\varphi((2n+1) \Delta)=K_{2}$ for
$n=0,\pm 1, \pm 2, \dots$ and is linear between these points.
Then, according to the generalized Ito formula,
$S$ satisfies
\begin{equation}
S_{t} = S_{0} + \int_{0}^{t} D_{-}\varphi(B_{u}) dB_{u} + \frac{1}{2}\int_{\mathbb{R}} L_{t}^{x}\varphi^{\prime \prime}(dx),
\end{equation}
where $D_{-}$ is the left-hand derivative operator,
$L_{\cdot}^{x}$ is the local time process of $B$ at $x$,
and $\varphi^{\prime \prime}$ is the signed measure associated with $\varphi$;
that is, $\varphi^{\prime \prime}=2(\sum_{k:\text{even}} \delta_{k \Delta}-\sum_{k:\text{odd}} \delta_{k \Delta})$ with $\delta_{x}$ the Dirac-delta measure at $x$.

This stock price is interpreted as the result
for the fundraiser to repurchase the stocks at time $t$ with $S_{t} = K_{1}$
and to issue the stocks at time $t$ with $S_{t} = K_{2}$
so that the number $n$ of the stock is
\begin{equation}
n_{t} = n_{0} \exp \left(\frac{1}{K_{2}}\sum_{k:\text{odd}} L_{t}^{k\Delta}-\frac{1}{K_{1}}\sum_{k:\text{even}} L_{t}^{k\Delta} \right).
\end{equation}
\end{example}

\subsubsection{In the presence of arbitragers}
\label{sec:in the presence of arbitragers}
In addition to the fundraiser,
we consider the stock price process in the presence of an arbitrager (``she'').
Let her trading strategy be denoted by $g \in L(S)$,
which is assumed to satisfy $g = 0$ on $\{t : d\kappa_{t} = 0\}$,
$g_{t}d\kappa_{t} \ge 0$,
and $g_{t} d\left<M,M\right>_{t}=0$.
Then, her gain process $G_{t} := \int_{0}^{t} g_{u} dS_{u} = \int_{0}^{t} g_{u} d\kappa_{u}$ is nondecreasing.

Instead of Assumption \ref{ass:constraint1},
we assume the following constraint on $F,G$, and $H$ in the presence of an arbitrager:
\begin{assumption}
The processes $F,G$, and $H$ satisfy $F_{t} + G_{t} + H_{t} = 0$ for all $t \ge 0$.
\end{assumption}

Under this constraint, the singular term $\kappa$ satisfies
\begin{equation}
d\kappa_{t} = -\frac{S_{t}}{n_{t}+g_{t}} dn_{t} = -\frac{S_{t}}{1+q_{t}} \frac{dn_{t}}{n_{t}} \label{eq:constraint-with-arb}
\end{equation}
with $q_{t}:=g_{t}/n_{t} \in (-1,\infty)$,
where $q_{t} > 1$ is interpreted to mean that the rest of market participants sell the stock short to her.
Then, the market capitalization $V_{t} = n_{t}S_{t}$ is 
\begin{equation}
V_{t} = V_{0} + \int_{0}^{t} n_{u}d(S_{u}-\kappa_{u}) -G_{t}. \label{eq:capital_witharb}
\end{equation}

The situation is more complicated than that in the previous subsection
in the sense that the funding and arbitrage strategies simultaneously affect the stock price.
Rather than finding an economic equilibrium, we consider a simplified problem.
For a given stock price process $S$ as \eqref{eq:underlying1d},
suppose that the arbitrager takes her strategies
such that $q$ is constant on $\{t : d\kappa_{t} > 0\}$
and on $\{t : d\kappa_{t} < 0\}$, respectively.
Suppose furthermore that the fundraiser takes his funding strategies
such that the number $n$ of the stocks is determined by \eqref{eq:constraint-with-arb}.
Then, the problem is whether her strategies are optimal to her or not in a certain sense.
We consider this problem using Examples \ref{ex:emery-perkins2} and \ref{ex:doubly-reflected-BM}.

\begin{example}
We consider the problem of Example \ref{ex:emery-perkins2} in the presence of an arbitrager with $q_{t} = q \in [0,\infty)$ and $\alpha > 1$ for simplicity.
The stock price follows \eqref{eq:emery-perkins2},
if the fundraiser repurchases the stocks such that the number of the stocks is
\begin{equation}
n_{t}^{(q)} = n_{0} \left(1+\frac{c}{b} L_{t} \right)^{-p}
\end{equation}
with $p = (1+q)(1+1/c) \ge 1$.
In this case, the money that he spends to repurchase the stocks up to the time $t$ is
\begin{equation}
-F_{t}^{(q)} = \frac{p}{p-1} n_{0}
\left(1-\left(1+\frac{c}{b}L_{t}\right)^{1-p}\right) b.
\end{equation}
Her gain up to the time $t$ is 
\begin{equation}
G_{t}^{(q)} = -\frac{p}{p+1}F_{t}
= \frac{p^{2}}{p^{2}-1} n_{0}
\left(1-\left(1+\frac{c}{b}L_{t}\right)^{1-p}\right) b.
\end{equation}

We consider the asymptotic case: $q \longrightarrow \infty$.
Let $\tau := \inf \{t: S_{t} = b \}$,
which is the first time for him to repurchase the stock.
Then, the number $n_{t}^{(q)}$ of the stocks is getting to $n_{0}1_{[0,\tau]}(t)$,
and her gain $G_{t}^{(q)}$ is to $n_{0}b1_{(\tau,\infty)}(t)$,
which means that the corporate exits the stock market at $\tau$,
and her gain at that time is the market capitalization $V_{\tau} = n_{0}b$.
\end{example}

\begin{example}
We consider the problem of Example \ref{ex:doubly-reflected-BM}
in the presence of an arbitrager
with $q_{t} = q_{1} > 0$ at the time $t$ with $S_{t}=K_{1}$,
$q_{t} = -q_{2} \in (-1,0)$ at the time $t$ with $S_{t}=K_{2}$ for some
constant $q_{1},q_{2}$.
The fundraiser has to issue and repurchase the stocks such that the number of the stocks is
\begin{equation}
n_{t}^{(q_{1},q_{2})} = n_{0} \exp \left(\frac{1-q_{2}}{K_{2}}\sum_{k:\text{odd}} L_{t}^{k\Delta}-\frac{1+q_{1}}{K_{1}}\sum_{k:\text{even}} L_{t}^{k\Delta} \right).
\end{equation}
As a result, the gain process $G^{(q_{1},q_{2})}$ of the arbitrager is 
\begin{align}
G_{t}^{(q_{1},q_{2})} &= \sum_{k} \int_{0}^{t} |q_{u}| n_{u}^{(q_{1},q_{2})} dL_{u}^{k\Delta} \\
&= n_{t}^{(q_{1},q_{2})} \left(\frac{q_{2}K_{2}}{1-q_{2}} \sum_{k:\text{odd}} \left(1-\mathrm{e}^{-(1-q_{2})L_{t}^{k\Delta}/K_{2}}\right)
+ \frac{q_{1}K_{1}}{1+q_{1}} \sum_{k:\text{even}} \left(\mathrm{e}^{(1+q_{1})L_{t}^{k\Delta}/K_{1}}-1\right)\right).
\end{align}

We consider the two extreme cases $q_{1} \longrightarrow \infty$ and $q_{2} \longrightarrow 1$,
assuming $B_{0} \in (0,\Delta)$ for simplicity.
Let $\tau_{i} := \inf \{t > 0: S_{t} = K_{i} \}$ for $i=1,2$.

We consider the asymptotic case: $q_{1} \longrightarrow \infty$.
As in the previous example,
the corporate exits the stock market at $\tau_{1}$,
at which time she makes profits $V_{\tau_{1}} = n_{\tau_{1}}K_{1}$.
Then, her gain process satisfies $G_{t}^{(\infty,q_{2})} \ge n_{\tau_{1}}K_{1} 1_{\{\tau_{1} < t\}}$.
In contrast to Example \ref{ex:emery-perkins2},
this strategy is not necessarily the optimal one.
Consider the arbitrage strategy
such that 
$q_{t} = 0$ for $t < \tau_{3} := \inf\{t > \tau_{2}:S_{t}=K_{1}\}$
and $q_{\tau_{3}} = \infty$.
Then, the corresponding gain process $G$ satisfies 
$G_{\tau_{3}} \ge n_{\tau_{3}} K_{1}$
with the corresponding number of the stocks $n$.
On the event $\{\tau_{1} < \tau_{2}\} \cap \{n_{0} < n_{\tau_{3}}\}$,
which has a positive probability,
the terminal gains $G_{\infty}^{(\infty,q_{2})}$ and $G_{\infty}$ of the two strategies satisfy
\begin{equation}
G_{\infty}^{(\infty,q_{2})}=G_{\tau_{1}}^{(\infty,q_{2})}=n_{0} K_{1}
< n_{\tau_{3}} K_{1} \le G_{\tau_{3}} = G_{\infty}.
\end{equation}

We consider the asymptotic case: $q_{2} \longrightarrow 1$.
In this case, we obtain 
\begin{align}
n_{t}^{(q_{1},1)} &= n_{0} \exp \left(-\frac{1+q_{1}}{K_{1}}\sum_{k:\text{even}} L_{t}^{k\Delta} \right),\\
G_{t}^{(q_{1},1)} &= n_{t}^{(q_{1},1)}\left(\sum_{k:\text{odd}} L_{t}^{k\Delta}
+ \frac{q_{1}}{1+q_{1}}K_{1} \sum_{k:\text{even}} \left(\mathrm{e}^{(1+q_{1})L_{t}^{k\Delta}/K_{1}}-1\right)\right),\\
F_{t}^{(q_{1},1)} &= n_{t}^{(q_{1},1)} \frac{q_{1}}{1+q_{1}}K_{1} \sum_{k:\text{even}} \left(\mathrm{e}^{(1+q_{1})L_{t}^{k\Delta}/K_{1}}-1\right).
\end{align}
The number $n_{t}^{(q_{1},1)}$ of the stocks does not change at time $t$ with $S_{t} = K_{2}$.
We interpret to mean that 
the fundraiser tries to raise money by issuing $dn_{t}^{(q_{1},1)}$-units of new stocks at time $t$ with $S_{t}=K_{2}$
and the arbitrager instantaneously sells $n_{t}^{(q_{1},1)}$-units of the stock short,
which leads him to refrain from the issuance.
As a result, he cannot raise money at all.
\end{example}

\subsection{Funding strategies without increasing profits}
\label{sec:fs_nip}
The presence of arbitragers is not desirable to the fundraiser
from the perspective of the market capitalization (see \eqref{eq:capital} and \eqref{eq:capital_witharb}).
This observation leads us to propose funding strategies that do not provide
increasing profits.

Suppose that
if the fundraiser did not take any funding strategies, the stock price process $S$ would be a diffusion
that satisfies
\begin{align}
S_{t} = S_{0} + \int_{0}^{t} \sigma(S_{u}) d\beta_{u} + \int_{0}^{t} \mu(S_{u})du, \label{sde:original}
\end{align}
where $0,\infty$ are assumed to be absorption boundaries,
$\beta$ is a Brownian motion with $\beta_{0} = 0$,
and $\sigma,\mu \in C^{1}((0,\infty))$ with $\sigma(y) \neq 0$ in $(0,\infty)$,
and that he makes decisions on funding strategies $n$ based on $\mu$ and $\sigma^{2}$ (the choice of the sign of $\sigma$ is up to him).
Then, his funding strategy $n$ introduces a singular term $\kappa$ into the semimartingale decomposition of \eqref{sde:original}.
This term is completely determined by $n$ in the absence of arbitragers
as seen in Section \ref{sec:zero-sum}.
Suppose that the resultant stock price process $Y$ satisfies
\begin{align}
Y_{t} = Y_{0} + \int_{0}^{t} \sigma(Y_{u}) d\beta_{u} + \int_{0}^{t} \tilde{\mu}(Y_{u})du + \kappa_{t}
\end{align}
for $Y_{0}=S_{0}$ and an appropriate $\tilde{\mu}$,
where $0$ is assumed to be an absorption boundary.
Hereafter, we denote $S$ as the hypothetical stock price process
and $(1,Y)$ as the resultant market model.
If he could choose his funding strategies $n$ such that
$\kappa$ is not adapted to the filtration generated by $Y$,
the resultant market models would not provide increasing profits
to agents who observe only the stock price $Y$.
This is possible under some conditions using Pitman's theorem (see Example \ref{ex:emery-perkins2} with $\alpha=2$ for a typical example).

\subsubsection{Funding strategies and the resultant stock price processes}
\label{sec:fs_and_process}
We propose two types of funding strategies: the repurchase and issuance strategies.
Intuitively, it would be better to repurchase the stock if $S_{t} \longrightarrow 0$,
and to issue if $S_{t} \longrightarrow \infty$.
Then, we consider the problem under the following assumption.
\begin{assumption}
\label{ass:transient}
The stock price process $S$ is not recurrent;
that is, $s(0)\neq -\infty$ or $s(\infty) \neq \infty$,
where $s$ is the scale function of $S$:
\begin{equation}
s(y) := \int_{Y_{0}}^{y} \exp \left(-2 \int_{Y_{0}}^{\eta} \frac{\mu(\xi)}{\sigma(\xi)^{2}} d\xi \right) d\eta.
\end{equation}
\end{assumption}

Let $\psi$ be the Lamperti-type transformation:
\begin{equation}
\psi(y) := \int_{Y_{0}}^{y} \frac{d\eta}{\sigma(\eta)},
\end{equation}
$h := f \circ \psi^{-1}$ and $T_{h}:=h^{\prime\prime}/h^{\prime}$,
where $f$ is a nonconstant affine transformation of $s$, which will be specified later.
We remark that $T_{h}$ is well defined, because it is invariant under an affine transformation.
Let $l$ and $r$ be the endpoints of the interval $\psi((0,\infty))$;
that is, $l := \min\{\psi(0),\psi(\infty)\}$ and $r := \max\{\psi(0),\psi(\infty)\}$.
The Lamperti-type transformation $\psi$ transforms $S$ to a semimartingale with a Brownian motion as a local martingale part:
\begin{equation}
\psi(S_{t}) = \psi(S_{0}) + \beta_{t} - \frac{1}{2} \int_{0}^{t} T_{h}(\psi(S_{u}))du.
\end{equation}
Then, \eqref{sde:original} is expressed with
\begin{align}
S_{t} = S_{0} + \int_{0}^{t} \sigma(S_{u}) d\beta_{u} -\frac{1}{2} \int_{0}^{t}T_{f}(S_{u})\sigma(S_{u})^{2}du.
\end{align}

\paragraph{the repurchase strategy}
We propose funding strategies under the condition $s(0) \neq -\infty$.
We assume that $\sigma > 0$, which makes the notation easier.
We define two increasing functions $f : (0,\infty) \longrightarrow \mathbb{R}$
and $h : (l,r) \longrightarrow \mathbb{R}_{+}$
by $f(y)=s(y)+C$ and  $h=f\circ \psi^{-1}$, respectively,
where $C$ is a constant such that $f(0)=s(0)+C = 0$ if $l = \psi(0) = -\infty$,
and $f(0)=s(0)+C \ge 0$, otherwise.
For a Brownian motion $\beta^{1/h}$ with $\beta_{0}^{1/h} \in (l,r)$ defined on some probability space,
let $X$ be a diffusion taking its values in $[l,r]$ that satisfies
\begin{equation}
X_{t} = \beta_{t}^{1/h} -\frac{1}{2} \int_{0}^{t} T_{1/h}(X_{u}) du. \label{eq:transientX}
\end{equation}
We assume that $X$ is a transient diffusion with $l < \inf_{t\ge0} X_{t}$ and $\sup_{t\ge0} X_{t} = r$,
or that $h$ can be extended to a function on $(l^{*},r)$ with $l^{*} < l$
such that $X$ is extended to be a transient diffusion
that satisfies \eqref{eq:transientX} with $l^{*} < \inf_{t\ge0} X_{t}$ and $\sup_{t\ge0} X_{t} = r$.
In the former case, let $l^{*} := l$.

Let $\mathcal{F}_{t}^{X,J}:=\mathcal{F}_{t}^{X} \vee \sigma(J_{t})$,
where $(\mathcal{F}_{t}^{X})_{t \ge 0}$ is the usual augmentation of the filtration generated by $X$
and $J_{t} := \inf_{u > t} X_{u}$.
Then, according to Aksamit and Jeanblanc \cite[Section 5.7]{aksamit2017enlargement},
a Brownian motion $\beta^{h}$ with respect to the enlarged filtration $(\mathcal{F}_{t}^{X,J})$ exists such that
\begin{equation}
X_{t} = \beta_{t}^{h} -\frac{1}{2} \int_{0}^{t} T_{h}(X_{u}) du + 2J_{t}.
\end{equation}
Let $\zeta := \inf\{t : X_{t} = l\}$ and $Y := \psi^{-1}(X^{\zeta})$.
Then, $Y$ has the following two semimartingale decompositions
with respect to $(\mathcal{F}_{t}^{X})$ and $(\mathcal{F}_{t}^{X,J})$, respectively:
\begin{align}
Y_{t} &= Y_{0} + \int_{0}^{t \wedge \zeta} \sigma(Y_{u}) d\beta_{u}^{1/h} -\frac{1}{2} \int_{0}^{t \wedge \zeta} T_{1/f}(Y_{u})\sigma(Y_{u})^{2}du \label{sde:withfunding_to_ordinary} \\
&= Y_{0} + \int_{0}^{t \wedge \zeta} \sigma(Y_{u}) d\beta_{u}^{h} -\frac{1}{2} \int_{0}^{t \wedge \zeta}T_{f}(Y_{u}) \sigma(Y_{u})^{2}du + \kappa_{t \wedge \zeta}, \label{sde:withfunding_to_fundraiser}
\end{align}
where $\kappa_{t} := 2(\inf_{u>t} Y_{u}-\inf_{u>0} Y_{u})$.

If the fundraiser repurchases the stocks such that
the number of the stocks is
\begin{equation}
n_{t} = n_{0}\exp \left(-\int_{0}^{t \wedge \zeta}\frac{d\kappa_{u}}{Y_{u}}\right)
= n_{0}\left(\frac{\psi^{-1}(J_{0})}{\psi^{-1}(J_{t})}\right)^{2} 1_{\{J_{0} > l\}}
+n_{0}1_{\{J_{0} \le l\}},
\label{eq:fs_nip}
\end{equation}
the stock price process has the same law as $Y$.
The representation \eqref{sde:withfunding_to_ordinary}
is the semimartingale decomposition with respect to the filtration $(\mathcal{F}_{t}^{X})$,
which expresses the information available to the ordinary investors.
In the resultant market model, no increasing profit exists.
On the contrary,
the representation \eqref{sde:withfunding_to_fundraiser}
is the semimartingale decomposition with respect to the filtration $(\mathcal{F}_{t}^{X,J})$,
which expresses the information available to the fundraiser.

\paragraph{the issuance strategy}
We propose funding strategies under the condition $s(\infty) \neq \infty$
and the assumption $\sigma < 0$.
We mention only the differences from the repurchasing strategies.
The function $f$ is defined as a decreasing function by $f(y)=-s(y)+C$,
where $C$ is a constant such that $f(\infty)=-s(\infty)+C = 0$ if $l = \psi(\infty) = -\infty$,
and $f(\infty)=-s(\infty)+C \ge 0$, otherwise.
The function $h$ is defined by $h=f\circ \psi^{-1}$ on $(l,r)$.
We make an assumption about $h$, which is extended if necessary,
so that that $X$ is a transient diffusion that satisfies
\eqref{eq:transientX} with $l^{*} < \inf_{t\ge0} X_{t}$ and $\sup_{t\ge0} X_{t} = r$ for some $l^{*} \le l$.
Then, $Y:= \psi^{-1}(X^{\zeta})$ satisfies \eqref{sde:withfunding_to_ordinary} and \eqref{sde:withfunding_to_fundraiser}
with $\kappa_{t} := 2(\sup_{u>t} Y_{u}-\sup_{u>0} Y_{u})$.
If the fundraiser issues the stocks such that \eqref{eq:fs_nip} is satisfied,
the stock price process has the same law as $Y$.

\vspace{\baselineskip}
Hereafter, we define $\Sigma_{t}^{Y}$
as $\Sigma_{t}^{Y} := 2(\inf_{u>t} Y_{u}-\inf_{u>0} Y_{u})$
for the repurchase strategy
and as $\Sigma_{t}^{Y} := 2(\sup_{u>t} Y_{u}-\sup_{u>0} Y_{u})$
for the issuance strategy.
\begin{example}
\label{ex:gbm}
We consider the case of $\sigma(y)^{2}=\sigma^{2} y^{2},\mu(y)=\mu y$ with $\sigma \neq 0$,
where the hypothetical stock price process \eqref{sde:original} is a geometric Brownian motion; that is,
\begin{equation}
S_{t} = S_{0} + \int_{0}^{t} \sigma S_{u} d\beta_{u} + \int_{0}^{t} \mu S_{u} du
= S_{0} \mathrm{e}^{\sigma \beta_{t} + (\mu-\sigma^{2}/2) t}.
\end{equation}
A straightforward computation leads to 
\begin{equation}
s(y) = \frac{Y_{0}}{\alpha} \left(\left(\frac{y}{Y_{0}}\right)^{\alpha}-1\right),
\end{equation}
for $\alpha := 1-2\mu/\sigma^{2} \neq 0$,
and $s(y) = Y_{0}\log \frac{y}{Y_{0}}$ for $\alpha=0$.

\begin{description}
\item[in case of $\alpha = 0$ :]
The fundraiser can take neither of the abovementioned two funding strategies.
\item[in case of $\alpha > 0$ :]
If the fundraiser repurchased the stocks according to the proposed strategy with
\begin{equation}
\psi(y) = \frac{1}{\sigma} \log \frac{y}{Y_{0}},\; h(x) = \frac{Y_{0}}{\alpha}\mathrm{e}^{\alpha\sigma x},\; f(y)=\frac{Y_{0}}{\alpha}\left(\frac{y}{Y_{0}}\right)^{\alpha},
\end{equation}
where we assume $\sigma > 0$,
the resultant stock price process $Y$ would be
\begin{align}
Y_{t} &= Y_{0} + \int_{0}^{t}\sigma Y_{u}d\beta_{u}^{1/h} + \int_{0}^{t}(\mu+\alpha \sigma^{2}) Y_{u}du \\
&= Y_{0} + \int_{0}^{t}\sigma Y_{u}d\beta_{u}^{h} + \int_{0}^{t} \mu Y_{u}du + \Sigma_{t}^{Y}.
\end{align}
\item[in case of $\alpha < 0$ :]
If the fundraiser issued the stocks according to the proposed strategy with
\begin{equation}
\psi(y) = \frac{1}{\sigma} \log \frac{y}{Y_{0}},\;
h(x) = -\frac{Y_{0}}{\alpha}\mathrm{e}^{\alpha\sigma x},\;
f(y)=-\frac{Y_{0}}{\alpha}\left(\frac{y}{Y_{0}}\right)^{\alpha},
\end{equation}
where we assume $\sigma < 0$,
the resultant stock price process $Y$ would be
\begin{align}
Y_{t} &= Y_{0} + \int_{0}^{t}\sigma Y_{u}d\beta_{u}^{1/h} + \int_{0}^{t} (\mu+\alpha \sigma^{2}) Y_{u}du \nonumber\\
&= Y_{0} + \int_{0}^{t}\sigma Y_{u}d\beta_{u}^{h} + \int_{0}^{t} \mu Y_{u}du + \Sigma_{t}^{Y}.
\end{align}
\end{description}
\end{example}

\begin{example}
\label{ex:repurchasing}
We consider the case of $\sigma(y)=1,\mu(y)=0$,
where the hypothetical stock price process \eqref{sde:original} is a Brownian motion stopped at $0$.
A straightforward computation leads to 
\begin{equation}
s(y) = y-Y_{0},\;
\psi(y) = y-Y_{0},\; h(x) = x + Y_{0} + C,\; f(y)=y+C
\end{equation}
with $C \ge 0$, $l=-Y_{0}$ and $r=\infty$.
Then, if the fundraiser repurchased the stocks according to the proposed strategy,
the stock price process $Y$ would satisfy
\begin{equation}
Y_{t} = Y_{0} + \beta_{t}^{1/h} + \int_{0}^{t} \frac{du}{Y_{u}+C}
= Y_{0} + \beta_{t}^{h}+ \Sigma_{t}^{Y}.
\end{equation}
In particular, $Y$ is a three-dimensional Bessel process,
which violates the NA condition, if $C=0$.
\end{example}

\begin{example}
We consider the case of $\sigma(y)=-y^{2},\mu(y)=y^{3}$,
where the hypothetical stock price process \eqref{sde:original} is
\begin{equation}
S_{t} = S_{0} -\int_{0}^{t}S_{u}^{2}d\beta_{u} + \int_{0}^{t}S_{u}^{3} du.
\end{equation}
We remark that $S$ explodes to $\infty$ in finite time.
A straightforward computation leads to 
\begin{equation}
s(y) = Y_{0}\left(1-\frac{Y_{0}}{y}\right),\;
\psi(y) = \frac{1}{y}-\frac{1}{Y_{0}},\;
h(x) = Y_{0}^{2}x+C,\;
f(y)=\frac{Y_{0}^{2}}{y}+C-Y_{0}
\end{equation}
with $C \ge Y_{0}$, $l=-1/Y_{0}$ and $r=\infty$.
Then, if the fundraiser issued the stocks according to the proposed strategy,
the stock price process $Y$ would follow
\begin{align}
Y_{t} &= Y_{0} -\int_{0}^{t}Y_{u}^{2}d\beta_{u}^{1/h}
 +\int_{0}^{t}Y_{u}^{3}\left(1-\frac{Y_{0}}{(C/Y_{0}-1)Y_{u}+Y_{0}}\right)du,\\
&= Y_{0} -\int_{0}^{t}Y_{u}^{2}d\beta_{u}^{h} + \int_{0}^{t}Y_{u}^{3} du
+ \Sigma_{t}^{Y}.
\end{align}
In particular, $Y$ is the reciprocal of a three-dimensional Bessel process, if $C=Y_{0}$.
\end{example}

\subsubsection{Funding, arbitrage opportunities, and bankruptcy}
\label{sec:bankruptcy}
The funding strategies proposed in Section \ref{sec:fs_and_process} affect the drift term of the stock price process,
which possibly introduces arbitrage opportunities.
In addition, 
these funding strategies can avoid or lead to {\it bankruptcy},
which is defined in this study
as the event that the corporate stock is worthless $S_{t}=0$.
For instance, Example \ref{ex:repurchasing} with $C=0$ demonstrates that the repurchase plan avoids bankruptcy and yields arbitrage.
In this section, we investigate conditions of the bankruptcy and arbitrage opportunities with and without funding strategies.

We consider three semimartingales: 
the hypothetical stock price process $S$,
the resultant stock price process $Y$,
and the hypothetical stock price process $S$
of \eqref{sde:original} with $\mu=0$.
These processes are assumed to be $(0,\infty)$- or $[0,\infty)$-valued regular diffusions,
where $0$ is assumed to be an absorption boundary if it is accessible,
and $\infty$ is inaccessible.
The laws of these processes as well as the corresponding market models with an arbitrary finite time horizon
are denoted as $P^{f},P^{1/f},Q$, respectively.

Criens and Urusov \cite{criens2024} provided deterministic criteria for the NA, NA1, and NFLVR conditions (see Criens and Urusov \cite[Theorems 3.9 and 3.10, and Corollary 3.11]{criens2024}) in general diffusion market models,
which is shortly reviewed in our specific case as follows.
These criteria consist of two conditions
that are involved with boundary points of the diffusion that is finite.
One criterion is involved with the instantaneous drift $\mu$ and diffusion terms $\sigma$:
\begin{equation}
\int_{0}^{1} yT_{f}(y)^{2} dy
= \int_{0}^{1} y\left(\frac{2\mu(y)}{\sigma(y)^{2}}\right)^{2} dy
< \infty. \label{condition:3.6}
\end{equation}
The other criterion is the accessibility to $0$ with respect to an appropriate measure.
The criteria for the financial market model $P^{f}$ are summarized as follows:
\begin{align}
\mathrm{NA} &\Longleftrightarrow
\text{\eqref{condition:3.6} holds, or $0$ is $Q$-inaccessible;}\\
\mathrm{NA1} &\Longleftrightarrow
\text{\eqref{condition:3.6} holds, or $0$ is $P^{f}$-inaccessible;}\\
\mathrm{NFLVR} &\Longleftrightarrow
\text{\eqref{condition:3.6} holds, or $0$ is $P^{f}$-and $Q$-inaccessible.}
\end{align}
The corresponding conditions for $P^{1/f}$ with funding for the ordinary investors
are obtained by replacing $f$ with $1/f$.
We remark that 
\begin{equation}
\frac{1}{4}\int_{\varepsilon}^{1} yT_{f}(y)^{2} dy
=\frac{1}{4}\int_{\varepsilon}^{1} yT_{1/f}(y)^{2} dy
+\log f(\varepsilon) - \varepsilon \frac{f^{\prime}(\varepsilon)}{f(\varepsilon)} + c \label{eq:dual}
\end{equation}
holds for $\varepsilon>0$, where $c$ is a constant independent of $\varepsilon$.

\paragraph{the repurchase strategy}
Suppose that $f(0)=0$ and $f^{\prime}>0$.
Then, the first assumption shows that $0$ is $P^{1/f}$-inaccessible
and the market model $P^{1/f}$ with funding satisfies the NA1 condition.
The second assumption together with \eqref{eq:dual}
shows that the condition \eqref{condition:3.6} for $1/f$ fails,
according to
\begin{equation}
\frac{1}{4}\int_{0}^{1} yT_{1/f}(y)^{2} dy
\ge -\lim_{\varepsilon \downarrow 0} \log f(\varepsilon)-c
= \infty.
\end{equation}
Then, the NFLVR condition fails in this market model,
in which case the NA condition fails,
if and only if
$0$ is $Q$-accessible; that is,
\begin{equation}
\int_{0}^{1} y \frac{dy}{\sigma(y)^{2}} = \infty \label{condition:3.7}
\end{equation}
fails.

Suppose further that $0$ is $P^{f}$-accessible; that is,
$\int_{0}^{1} \frac{f(y)}{f^{\prime}(y)} \frac{dy}{\sigma(y)^{2}} < \infty$,
which means that the corporate possibly goes bankrupt in the original market model $P^{f}$.
Then, the fundraiser would repurchase the stock.
In this case, he introduces arbitrage opportunities into
the market if and only if the condition \eqref{condition:3.7} fails.
For example, 
the NFLVR condition is satisfied in the market model in Example \ref{ex:besq} with $\delta < 2$ below.

\paragraph{the issuance strategy}
Suppose that $f(0)=\infty$ and $f^{\prime}<0$.
Then, the first assumption shows that $0$ is $P^{f}$-inaccessible
and the market model $P^{f}$ satisfies the NA1 condition.
The second assumption together with \eqref{eq:dual}
shows that the condition \eqref{condition:3.6} for $f$ fails,
according to
\begin{equation}
\frac{1}{4}\int_{0}^{1} yT_{f}(y)^{2} dy
\ge \lim_{\varepsilon \downarrow 0} \log f(\varepsilon) + c
= \infty.
\end{equation}
Then, the NFLVR condition fails if and only if the condition \eqref{condition:3.7} fails, in which case  the NA condition fails.

The fundraiser would issue new stocks if the funding did not yield the possibility of bankrupt, which is equivalent to
$\int_{0}^{1} \frac{1/f(y)}{(1/f)^{\prime}(y)} \frac{dy}{\sigma(y)^{2}} = \infty$.
In this case, the market model $P^{1/f}$ satisfies the NA1 condition.
In addition, this satisfies the NA condition
if and only if 
the condition \eqref{condition:3.6} for $1/f$ 
or \eqref{condition:3.7} is satisfied.
For example, the NFLVR condition is satisfied in the market model $P^{1/f}$ in Example \ref{ex:gbm} with $\alpha < 0$
and not in Example \ref{ex:besq} with $\delta > 2$ below.
In the latter case, however, the bankruptcy possibly occurs.

\vspace{\baselineskip}
\begin{example}
\label{ex:besq}
We consider the case of $\sigma(y)^{2}=(2\sqrt{y})^{2},\mu(y)=\delta \in \mathbb{R} \setminus \{2\}$,
where $S$ is a squared Bessel process with dimension $\delta$:
\begin{align}
S_{t} = S_{0} +2 \int_{0}^{t} \sqrt{S_{u}}d\beta_{u} + \delta t,
\end{align}
and $Y$ is that with dimension $4-\delta$.
We have omitted the case $\delta = 2$, where $S$ is recurrent.
\begin{description}
\item[in case of $\delta < 2$:]
Without repurchasing the stocks,
the stock price $S$ becomes $0$ with finite time with a positive probability.
Then, the fundraiser would decide to repurchase the stocks to avoid a bankruptcy.
As a result, the resultant stock price process $Y$ would be
a squared Bessel process with dimension $4-\delta > 2$.
He could avoid a bankruptcy without introducing arbitrage opportunities; that is, the NFLVR condition holds.
\item[in case of $\delta > 2$:]
If the fundraiser issued new stock according to the proposed funding strategies,
the resultant stock price process $Y$ would be a squared Bessel process with dimension $4-\delta < 2$.
Then, the issuance could lead to a bankruptcy.
\end{description}
\end{example}

\section{Conclusion}
\label{sec:conclusion}
We studied contingent claim valuation under increasing profit, strong arbitrage, and arbitrage of the first kind.
The pricing formula \eqref{eq:main} is an extension of that proposed in the literature for the market model with martingale deflators.
In particular,
this formula corrects the published option pricing formulae,
which are based on the mistaken belief that an ELMM exists.
This correction applies to the financial market model in which the stock price follows a reflected geometric Brownian motion.

Furthermore, we considered financial market models with increasing profit in the context of corporate stock issuance and repurchase plans.
According to the literature, an equilibrium asset price cannot exist in the presence of increasing profit.
However, we consider that economic equilibrium can be derived in such a market model.
This issue is left to future study.

%\bibliographystyle{apalike}
%\bibliographystyle{spmpsci}
%\bibliographystyle{plainnat}
%\bibliographystyle{plain}
%\bibliography{ref}   % name your BibTeX data base

\end{document}